\documentclass[12pt,english]{article}
\usepackage[T1]{fontenc}
\usepackage[latin9]{inputenc}
\usepackage{geometry}
\usepackage{array}
\usepackage{rotating}
\usepackage{float}
\usepackage{url}
\usepackage{multirow}
\usepackage{amsthm}
\usepackage{amsmath}
\usepackage{amssymb}
\usepackage{graphicx}

\makeatletter

\providecommand{\tabularnewline}{\\}

 \theoremstyle{definition}
 \newtheorem*{defn*}{\protect\definitionname}
  \theoremstyle{remark}
  \newtheorem*{rem*}{\protect\remarkname}
\theoremstyle{plain}
\newtheorem{thm}{\protect\theoremname}
  \theoremstyle{plain}
  \newtheorem{prop}[thm]{\protect\propositionname}

\usepackage{amsmath}
\usepackage{amssymb}
\usepackage{url}
\usepackage{listings}
\usepackage{url}
\@ifundefined{textcolor}{\usepackage{color}}{}

\makeatother

\usepackage{babel}
  \providecommand{\definitionname}{Definition}
  \providecommand{\propositionname}{Proposition}
  \providecommand{\remarkname}{Remark}
\providecommand{\theoremname}{Theorem}

\begin{document}

\title{Representation, simplification and display of fractional powers of
rational numbers in computer algebra}

\author{Albert D. Rich%
\thanks{Albert\_Rich at msn dot com %
}\\
David R. Stoutemyer%
\thanks{dstout at hawaii dot edu%
}}
\maketitle
\begin{abstract}
Simplification of fractional powers of positive rational numbers and
of sums, products and powers of such numbers is taught in beginning
algebra. Such numbers can often be expressed in many ways, as this
article discusses in some detail. Since they are such a restricted
subset of algebraic numbers, it might seem that good simplification
of them must already be implemented in all widely used computer algebra
systems. However, the algorithm taught in beginning algebra uses integer
factorization, which can consume unacceptable time for the large numbers
that often arise within computer algebra. Therefore some systems apparently
use various \emph{ad hoc} techniques that can return an incorrect
result because of not simplifying to 0 the difference between two
equivalent such expressions. Even systems that avoid this flaw often
do not return the same result for all equivalent such input forms,
or return an unnecessarily bulky result that does not have any other
compensating useful property. This article identifies some of these
deficiencies, then describes the advantages and disadvantages of various
alternative forms and how to overcome the deficiencies without costly
integer factorization.
\end{abstract}

\section{Why discuss such an elementary topic here?}

First:\vspace{-0.3em}

\begin{defn*}
An \textsl{absurd numbe}r is one that can be expressed as a rational
number times a product of zero or more fractional powers of positive
rational numbers.\vspace{-0.2em}
\end{defn*}
\begin{rem*}
We need a \textsl{brief} name for this subset of algebraic numbers,
and the inspiration for this one is that \emph{ab} means ``from''
in Latin, and ``absurd numbers'' continues the tradition started
with whimsical names such as surds, imaginary numbers, radicals, irrational
numbers, and surreal numbers.
\end{rem*}
This article discusses the advantages and disadvantages of alternative
ways computer algebra systems can represent, simplify, and display
absurd numbers. Although this is a topic taught in beginning algebra,
some major computer algebra systems do an imperfect or surprisingly
poor job; and we have suggestions for remedies.

\subsection{Simplification of equivalent forms of an absurd number }

Table \ref{SystemComparison} shows the results produced by four systems
for sixteen different input representations of the same absurd number.\vspace{2em}

\begin{table}[H]
\caption{Results for simplifying 16 input representations of the same absurd
number}
\label{SystemComparison}

\begin{tabular}{|c|c|c|c|c|c|c|c|c|}
\hline 
\multirow{2}{*}{{\scriptsize \negthinspace{}\negthinspace{}\negthinspace{}\#\negthinspace{}\negthinspace{}\negthinspace{}}} & \multirow{2}{*}{Input} & \textsl{\footnotesize \negthinspace{}\negthinspace{}Derive\negthinspace{}\negthinspace{}} & \textsl{\footnotesize \negthinspace{}\negthinspace{}\negthinspace{}Mathematica\negthinspace{}\negthinspace{}\negthinspace{}} & \multicolumn{2}{c|}{{\footnotesize Maple}} & \multicolumn{3}{c|}{{\footnotesize Maxima}}\tabularnewline
\cline{3-9} 
 &  & \textsl{\footnotesize \negthinspace{}\negthinspace{}}{\footnotesize default}\textsl{\footnotesize \negthinspace{}\negthinspace{}} & {\footnotesize default} & {\footnotesize default} & {\footnotesize simplify()} & {\footnotesize default} & \negthinspace{}\negthinspace{}{\tiny rootscontract()}\negthinspace{}\negthinspace{}\textsl{\footnotesize \negthinspace{}} & \negthinspace{}\negthinspace{}{\footnotesize radcan()}\negthinspace{}\negthinspace{}\tabularnewline
\hline 
\hline 
\multicolumn{2}{c|}{\quad{}prime bases} & \multicolumn{7}{c}{}\tabularnewline
\hline 
{\scriptsize 1} & \rule[-9pt]{0pt}{26pt}$\frac{2^{4/3}\,7^{2/3}}{3^{2/3}\,5^{2/3}}$ & \textsl{\footnotesize \negthinspace{}\negthinspace{}}$\frac{2\cdot1470^{1/3}}{15}$\textsl{\footnotesize \negthinspace{}\negthinspace{}} & \textsl{\footnotesize \negthinspace{}\negthinspace{}}${\scriptstyle 2\,\left({\textstyle \frac{7}{15}}\right)^{2/3}2^{1/3}}$\textsl{\footnotesize \negthinspace{}\negthinspace{}} & \negthinspace{}\negthinspace{}\negthinspace{}$\frac{2}{15}{\scriptstyle \,2^{1/3}\,3^{1/3}\,5^{1/3}\,7^{2/3}}$\negthinspace{}\negthinspace{}\negthinspace{} & \negthinspace{}\negthinspace{}\negthinspace{}$\frac{2}{15}{\scriptstyle \,2^{1/3}\,3^{1/3}\,5^{1/3}\,7^{2/3}}$\negthinspace{}\negthinspace{}\negthinspace{} & $\frac{2^{4/3}\,7^{2/3}}{3^{2/3}\,5^{2/3}}$ & $\boldsymbol{\frac{2\:98^{1/3}}{225^{1/3}}}$ & \negthinspace{}\negthinspace{}$\frac{2^{4/3}\,7^{2/3}}{3^{2/3}\,5^{2/3}}$\negthinspace{}\negthinspace{}\tabularnewline
\hline 
{\scriptsize 2} & \rule[-9pt]{0pt}{26pt}\negthinspace{}\negthinspace{}\negthinspace{}$\frac{2}{15}{\scriptstyle \,2^{1/3}\,3^{1/3}\,5^{1/3}\,7^{2/3}}$\negthinspace{}\negthinspace{}\negthinspace{} & \textsl{\footnotesize \negthinspace{}\negthinspace{}}$\frac{2\cdot1470^{1/3}}{15}$\textsl{\footnotesize \negthinspace{}\negthinspace{}} & \textsl{\footnotesize \negthinspace{}\negthinspace{}}${\scriptstyle 2\,\left({\textstyle \frac{7}{15}}\right)^{2/3}2^{1/3}}$\textsl{\footnotesize \negthinspace{}\negthinspace{}} & \negthinspace{}\negthinspace{}\negthinspace{}$\frac{2}{15}{\scriptstyle \,2^{1/3}\,3^{1/3}\,5^{1/3}\,7^{2/3}}$\negthinspace{}\negthinspace{}\negthinspace{} & \negthinspace{}\negthinspace{}\negthinspace{}$\frac{2}{15}{\scriptstyle \,2^{1/3}\,3^{1/3}\,5^{1/3}\,7^{2/3}}$\negthinspace{}\negthinspace{}\negthinspace{} & $\frac{2^{4/3}\,7^{2/3}}{3^{2/3}\,5^{2/3}}$ & $\boldsymbol{\frac{2\:98^{1/3}}{225^{1/3}}}$ & \negthinspace{}\negthinspace{}$\frac{2^{4/3}\,7^{2/3}}{3^{2/3}\,5^{2/3}}$\negthinspace{}\negthinspace{}\tabularnewline
\hline 
{\scriptsize 3} & \rule[-9pt]{0pt}{26pt}$\frac{14\cdot2^{1/3}\,3^{1/3}\,5^{1/3}}{15\cdot7^{1/3}}$ & \textsl{\footnotesize \negthinspace{}\negthinspace{}}$\frac{2\cdot1470^{1/3}}{15}$\textsl{\footnotesize \negthinspace{}\negthinspace{}} & \textsl{\footnotesize \negthinspace{}\negthinspace{}}${\scriptstyle 2\,\left({\textstyle \frac{7}{15}}\right)^{2/3}2^{1/3}}$\textsl{\footnotesize \negthinspace{}\negthinspace{}} & \negthinspace{}\negthinspace{}\negthinspace{}$\frac{2}{15}{\scriptstyle \,2^{1/3}\,3^{1/3}\,5^{1/3}\,7^{2/3}}$\negthinspace{}\negthinspace{}\negthinspace{} & \negthinspace{}\negthinspace{}\negthinspace{}$\frac{2}{15}{\scriptstyle \,2^{1/3}\,3^{1/3}\,5^{1/3}\,7^{2/3}}$\negthinspace{}\negthinspace{}\negthinspace{} & $\frac{2^{4/3}\,7^{2/3}}{3^{2/3}\,5^{2/3}}$ & $\boldsymbol{\frac{2\:98^{1/3}}{225^{1/3}}}$ & \negthinspace{}\negthinspace{}$\frac{2^{4/3}\,7^{2/3}}{3^{2/3}\,5^{2/3}}$\negthinspace{}\negthinspace{}\tabularnewline
\hline 
{\scriptsize 4} & \rule[-9pt]{0pt}{26pt}$\frac{2\cdot2^{1/3}\,7^{2/3}}{3^{2/3}\,5^{2/3}}$ & \textsl{\footnotesize \negthinspace{}\negthinspace{}}$\frac{2\cdot1470^{1/3}}{15}$\textsl{\footnotesize \negthinspace{}\negthinspace{}} & \textsl{\footnotesize \negthinspace{}\negthinspace{}}${\scriptstyle 2\,\left({\textstyle \frac{7}{15}}\right)^{2/3}2^{1/3}}$\textsl{\footnotesize \negthinspace{}\negthinspace{}} & \negthinspace{}\negthinspace{}\negthinspace{}$\frac{2}{15}{\scriptstyle \,2^{1/3}\,3^{1/3}\,5^{1/3}\,7^{2/3}}$\negthinspace{}\negthinspace{}\negthinspace{} & \negthinspace{}\negthinspace{}\negthinspace{}$\frac{2}{15}{\scriptstyle \,2^{1/3}\,3^{1/3}\,5^{1/3}\,7^{2/3}}$\negthinspace{}\negthinspace{}\negthinspace{} & $\frac{2^{4/3}\,7^{2/3}}{3^{2/3}\,5^{2/3}}$ & $\boldsymbol{\frac{2\:98^{1/3}}{225^{1/3}}}$ & \negthinspace{}\negthinspace{}$\frac{2^{4/3}\,7^{2/3}}{3^{2/3}\,5^{2/3}}$\negthinspace{}\negthinspace{}\tabularnewline
\hline 
\multicolumn{2}{c|}{{\footnotesize coprime square free}} & \multicolumn{7}{c}{}\tabularnewline
\hline 
{\scriptsize 5} & \rule[-9pt]{0pt}{26pt}$\frac{2^{4/3}\,7^{2/3}}{15^{2/3}}$ & \textsl{\footnotesize \negthinspace{}\negthinspace{}}$\frac{2\cdot1470^{1/3}}{15}$\textsl{\footnotesize \negthinspace{}\negthinspace{}} & \textsl{\footnotesize \negthinspace{}\negthinspace{}}${\scriptstyle 2\,\left({\textstyle \frac{7}{15}}\right)^{2/3}2^{1/3}}$\textsl{\footnotesize \negthinspace{}\negthinspace{}} & $\boldsymbol{\frac{2}{15}{\scriptstyle \,2^{1/3}\,7^{2/3}\,15^{1/3}}}$ & $\boldsymbol{\frac{2}{15}{\scriptstyle \,2^{1/3}\,7^{2/3}\,15^{1/3}}}$ & $\frac{2^{4/3}\,7^{2/3}}{15^{2/3}}$ & $\boldsymbol{\frac{2\:98^{1/3}}{225^{1/3}}}$ & \negthinspace{}\negthinspace{}$\frac{2^{4/3}\,7^{2/3}}{3^{2/3}\,5^{2/3}}$\negthinspace{}\negthinspace{}\tabularnewline
\hline 
{\scriptsize 6} & \rule[-9pt]{0pt}{26pt}$\frac{2}{15}{\scriptstyle \,7^{2/3}\,30^{1/3}}$ & \textsl{\footnotesize \negthinspace{}\negthinspace{}}$\frac{2\cdot1470^{1/3}}{15}$\textsl{\footnotesize \negthinspace{}\negthinspace{}} & \textsl{\footnotesize \negthinspace{}\negthinspace{}}${\scriptstyle 2\,\left({\textstyle \frac{7}{15}}\right)^{2/3}2^{1/3}}$\textsl{\footnotesize \negthinspace{}\negthinspace{}} & $\frac{2}{15}{\scriptstyle \,7^{2/3}\,30^{1/3}}$ & $\frac{2}{15}{\scriptstyle \,7^{2/3}\,30^{1/3}}$ & \negthinspace{}\negthinspace{}$\frac{2\:7^{2/3}30^{1/3}}{15}$\negthinspace{}\negthinspace{} & $\frac{2\:1470^{1/3}}{15}$ & \negthinspace{}\negthinspace{}$\frac{2^{4/3}\,7^{2/3}}{3^{2/3}\,5^{2/3}}$\negthinspace{}\negthinspace{}\tabularnewline
\hline 
{\scriptsize 7} & \rule[-9pt]{0pt}{26pt}$\frac{14\cdot30^{1/3}}{15\cdot7^{1/3}}$ & \textsl{\footnotesize \negthinspace{}\negthinspace{}}$\frac{2\cdot1470^{1/3}}{15}$\textsl{\footnotesize \negthinspace{}\negthinspace{}} & \textsl{\footnotesize \negthinspace{}\negthinspace{}}${\scriptstyle 2\,\left({\textstyle \frac{7}{15}}\right)^{2/3}2^{1/3}}$\textsl{\footnotesize \negthinspace{}\negthinspace{}} & $\frac{2}{15}{\scriptstyle \,7^{2/3}\,30^{1/3}}$ & $\frac{2}{15}{\scriptstyle \,7^{2/3}\,30^{1/3}}$ & \negthinspace{}\negthinspace{}$\frac{2\:7^{2/3}30^{1/3}}{15}$\negthinspace{}\negthinspace{} & $\frac{2\:1470^{1/3}}{15}$ & \negthinspace{}\negthinspace{}$\frac{2^{4/3}\,7^{2/3}}{3^{2/3}\,5^{2/3}}$\negthinspace{}\negthinspace{}\tabularnewline
\hline 
{\scriptsize 8} & \rule[-9pt]{0pt}{26pt}$\frac{2\cdot2^{1/3}\,7^{2/3}}{15^{2/3}}$ & \textsl{\footnotesize \negthinspace{}\negthinspace{}}$\frac{2\cdot1470^{1/3}}{15}$\textsl{\footnotesize \negthinspace{}\negthinspace{}} & \textsl{\footnotesize \negthinspace{}\negthinspace{}}${\scriptstyle 2\,\left({\textstyle \frac{7}{15}}\right)^{2/3}2^{1/3}}$\textsl{\footnotesize \negthinspace{}\negthinspace{}} & $\boldsymbol{\frac{2}{15}{\scriptstyle \,2^{1/3}\,7^{2/3}\,15^{1/3}}}$ & $\boldsymbol{\frac{2}{15}{\scriptstyle \,2^{1/3}\,7^{2/3}\,15^{1/3}}}$ & $\frac{2^{4/3}\,7^{2/3}}{15^{2/3}}$ & $\boldsymbol{\frac{2\:98^{1/3}}{225^{1/3}}}$ & \negthinspace{}\negthinspace{}$\frac{2^{4/3}\,7^{2/3}}{3^{2/3}\,5^{2/3}}$\negthinspace{}\negthinspace{}\tabularnewline
\hline 
{\scriptsize 9} & \rule[-9pt]{0pt}{26pt}${\scriptstyle 2\cdot2^{1/3}}\left(\frac{7}{15}\right)^{2/3}$ & \textsl{\footnotesize \negthinspace{}\negthinspace{}}$\frac{2\cdot1470^{1/3}}{15}$\textsl{\footnotesize \negthinspace{}\negthinspace{}} & \textsl{\footnotesize \negthinspace{}\negthinspace{}}${\scriptstyle 2\,\left({\textstyle \frac{7}{15}}\right)^{2/3}2^{1/3}}$\textsl{\footnotesize \negthinspace{}\negthinspace{}} & $\boldsymbol{\frac{2}{15}{\scriptstyle \,2^{1/3}\,7^{2/3}\,15^{1/3}}}$ & $\boldsymbol{\frac{2}{15}{\scriptstyle \,2^{1/3}\,7^{2/3}\,15^{1/3}}}$ & $\frac{2^{4/3}\,7^{2/3}}{15^{2/3}}$ & $\boldsymbol{\frac{2\:98^{1/3}}{225^{1/3}}}$ & \negthinspace{}\negthinspace{}$\frac{2^{4/3}\,7^{2/3}}{3^{2/3}\,5^{2/3}}$\negthinspace{}\negthinspace{}\tabularnewline
\hline 
{\scriptsize \negthinspace{}\negthinspace{}\negthinspace{}10\negthinspace{}\negthinspace{}\negthinspace{}} & \rule[-11pt]{0pt}{33pt}$\frac{14}{15}\left(\frac{30}{7}\right)^{1/3}$ & \textsl{\footnotesize \negthinspace{}\negthinspace{}}$\frac{2\cdot1470^{1/3}}{15}$\textsl{\footnotesize \negthinspace{}\negthinspace{}} & \textsl{\footnotesize \negthinspace{}\negthinspace{}}${\scriptstyle 2\,\left({\textstyle \frac{7}{15}}\right)^{2/3}2^{1/3}}$\textsl{\footnotesize \negthinspace{}\negthinspace{}} & $\frac{2}{15}{\scriptstyle \,7^{2/3}\,30^{1/3}}$ & $\frac{2}{15}{\scriptstyle \,7^{2/3}\,30^{1/3}}$ & \negthinspace{}\negthinspace{}$\frac{2\:7^{2/3}30^{1/3}}{15}$\negthinspace{}\negthinspace{} & $\frac{2\:1470^{1/3}}{15}$ & \negthinspace{}\negthinspace{}$\frac{2^{4/3}\,7^{2/3}}{3^{2/3}\,5^{2/3}}$\negthinspace{}\negthinspace{}\tabularnewline
\hline 
\multicolumn{2}{c|}{{\small non perfect powers}} & \multicolumn{7}{c}{}\tabularnewline
\hline 
{\scriptsize \negthinspace{}\negthinspace{}\negthinspace{}11\negthinspace{}\negthinspace{}\negthinspace{}} & \rule[-9pt]{0pt}{26pt}$\left(\frac{28}{15}\right)^{2/3}$ & \textsl{\footnotesize \negthinspace{}\negthinspace{}}$\frac{2\cdot1470^{1/3}}{15}$\textsl{\footnotesize \negthinspace{}\negthinspace{}} & \textsl{\footnotesize \negthinspace{}\negthinspace{}}${\scriptstyle 2\,\left({\textstyle \frac{7}{15}}\right)^{2/3}2^{1/3}}$\textsl{\footnotesize \negthinspace{}\negthinspace{}} & $\frac{1}{15}{\scriptstyle 28^{2/3}15^{1/3}}$ & $\frac{1}{15}{\scriptstyle 28^{2/3}15^{1/3}}$ & $\frac{\boldsymbol{28^{2/3}}}{\boldsymbol{15^{2/3}}}$ & $\boldsymbol{\frac{2\:98^{1/3}}{225^{1/3}}}$ & \negthinspace{}\negthinspace{}$\frac{2^{4/3}\,7^{2/3}}{3^{2/3}\,5^{2/3}}$\negthinspace{}\negthinspace{}\tabularnewline
\hline 
{\scriptsize \negthinspace{}\negthinspace{}\negthinspace{}12\negthinspace{}\negthinspace{}\negthinspace{}} & \rule[-9pt]{0pt}{26pt}$\frac{28^{2/3}}{15^{2/3}}$ & \textsl{\footnotesize \negthinspace{}\negthinspace{}}$\frac{2\cdot1470^{1/3}}{15}$\textsl{\footnotesize \negthinspace{}\negthinspace{}} & \textsl{\footnotesize \negthinspace{}\negthinspace{}}${\scriptstyle 2\,\left({\textstyle \frac{7}{15}}\right)^{2/3}2^{1/3}}$\textsl{\footnotesize \negthinspace{}\negthinspace{}} & $\frac{1}{15}{\scriptstyle \,28^{2/3}\,15^{1/3}}$ & $\frac{1}{15}{\scriptstyle \,28^{2/3}\,15^{1/3}}$ & $\frac{\boldsymbol{28^{2/3}}}{\boldsymbol{15^{2/3}}}$ & $\boldsymbol{\frac{2\:98^{1/3}}{225^{1/3}}}$ & \negthinspace{}\negthinspace{}$\frac{2^{4/3}\,7^{2/3}}{3^{2/3}\,5^{2/3}}$\negthinspace{}\negthinspace{}\tabularnewline
\hline 
\multicolumn{2}{c|}{\negthinspace{}\negthinspace{}\negthinspace{}\negthinspace{}{\footnotesize max
reciprocal powers}\negthinspace{}\negthinspace{}\negthinspace{}\negthinspace{}} & \multicolumn{7}{c}{}\tabularnewline
\hline 
{\scriptsize \negthinspace{}\negthinspace{}\negthinspace{}13\negthinspace{}\negthinspace{}\negthinspace{}} & \rule[-9pt]{0pt}{26pt}$\left(\frac{784}{225}\right)^{1/3}$ & \textsl{\footnotesize \negthinspace{}\negthinspace{}}$\frac{2\cdot1470^{1/3}}{15}$\textsl{\footnotesize \negthinspace{}\negthinspace{}} & \textsl{\footnotesize \negthinspace{}\negthinspace{}}${\scriptstyle 2\,\left({\textstyle \frac{7}{15}}\right)^{2/3}2^{1/3}}$\textsl{\footnotesize \negthinspace{}\negthinspace{}} & $\boldsymbol{\frac{1}{225}{\scriptstyle \,784^{1/3}\,225^{2/3}}}$ & $\boldsymbol{\frac{2}{15}{\scriptstyle \,98^{1/3}\,15^{1/3}}}$ & $\boldsymbol{\frac{2\:98^{1/3}}{225^{1/3}}}$ & $\boldsymbol{\frac{2\:98^{1/3}}{225^{1/3}}}$ & \negthinspace{}\negthinspace{}$\frac{2^{4/3}\,7^{2/3}}{3^{2/3}\,5^{2/3}}$\negthinspace{}\negthinspace{}\tabularnewline
\hline 
{\scriptsize \negthinspace{}\negthinspace{}\negthinspace{}14\negthinspace{}\negthinspace{}\negthinspace{}} & \rule[-9pt]{0pt}{26pt}$\frac{784^{1/3}}{225^{1/3}}$ & \textsl{\footnotesize \negthinspace{}\negthinspace{}}$\frac{2\cdot1470^{1/3}}{15}$\textsl{\footnotesize \negthinspace{}\negthinspace{}} & \textsl{\footnotesize \negthinspace{}\negthinspace{}}${\scriptstyle 2\,\left({\textstyle \frac{7}{15}}\right)^{2/3}2^{1/3}}$\textsl{\footnotesize \negthinspace{}\negthinspace{}} & $\boldsymbol{\frac{1}{225}{\scriptstyle \,784^{1/3}\,225^{2/3}}}$ & $\boldsymbol{\frac{2}{15}{\scriptstyle \,98^{1/3}\,15^{1/3}}}$ & $\boldsymbol{\frac{2\:98^{1/3}}{225^{1/3}}}$ & $\boldsymbol{\frac{2\:98^{1/3}}{225^{1/3}}}$ & \negthinspace{}\negthinspace{}$\frac{2^{4/3}\,7^{2/3}}{3^{2/3}\,5^{2/3}}$\negthinspace{}\negthinspace{}\tabularnewline
\hline 
\multicolumn{2}{c|}{one integer power} & \multicolumn{7}{c}{}\tabularnewline
\hline 
{\footnotesize \negthinspace{}\negthinspace{}\negthinspace{}}{\scriptsize 15}{\footnotesize \negthinspace{}\negthinspace{}\negthinspace{}} & \rule[-9pt]{0pt}{26pt}$\frac{2}{15}{\scriptstyle \,1470^{1/3}}$ & \textsl{\footnotesize \negthinspace{}\negthinspace{}}$\frac{2\cdot1470^{1/3}}{15}$\textsl{\footnotesize \negthinspace{}\negthinspace{}} & \textsl{\footnotesize \negthinspace{}\negthinspace{}}${\scriptstyle 2\,\left({\textstyle \frac{7}{15}}\right)^{2/3}2^{1/3}}$\textsl{\footnotesize \negthinspace{}\negthinspace{}} & $\frac{2}{15}{\scriptstyle \,1470^{1/3}}$ & $\frac{2}{15}{\scriptstyle \,1470^{1/3}}$ & $\frac{2\:1470^{1/3}}{15}$ & $\frac{2\:1470^{1/3}}{15}$ & \negthinspace{}\negthinspace{}$\frac{2^{4/3}\,7^{2/3}}{3^{2/3}\,5^{2/3}}$\negthinspace{}\negthinspace{}\tabularnewline
\hline 
{\footnotesize \negthinspace{}\negthinspace{}\negthinspace{}}{\scriptsize 16\negthinspace{}}{\footnotesize \negthinspace{}\negthinspace{}} & \rule[-9pt]{0pt}{26pt}$\frac{1}{15}{\scriptstyle \,11760^{1/3}}$ & \textsl{\footnotesize \negthinspace{}\negthinspace{}}$\frac{2\cdot1470^{1/3}}{15}$\textsl{\footnotesize \negthinspace{}\negthinspace{}} & \textsl{\footnotesize \negthinspace{}\negthinspace{}}${\scriptstyle 2\,\left({\textstyle \frac{7}{15}}\right)^{2/3}2^{1/3}}$\textsl{\footnotesize \negthinspace{}\negthinspace{}} & $\frac{1}{15}{\scriptstyle \,11760^{1/3}}$ & $\frac{2}{15}{\scriptstyle \,1470^{1/3}}$ & $\frac{2\:1470^{1/3}}{15}$ & $\frac{2\:1470^{1/3}}{15}$ & \negthinspace{}\negthinspace{}$\frac{2^{4/3}\,7^{2/3}}{3^{2/3}\,5^{2/3}}$\negthinspace{}\negthinspace{}\tabularnewline
\hline 
\end{tabular}
\end{table}

Regarding columns labeled ``default'':
\begin{defn*}
\textsl{Default simplification} is the result of pressing \fbox{\sf{Enter}}
in Maple, \fbox{\sf{Ctrl}} \fbox{\sf{Enter}} in \textsl{Derive},
or \fbox{\sf{Shift}} \fbox{\sf{Enter}} in \textsl{Mathematica} or
wxMaxima \textbf{--} with the factory-default mode settings and no
transformational or simplification functions anywhere in the input
expression.
\end{defn*}
For Maple 15, $\mathrm{simplify}(\ldots,\mathrm{size})$ gave some
different results than $\mathrm{simplify}(\ldots)$ \textbf{--} not
always smaller in terms of any easily discerned measure. For Maxima
5.24, the $\mathrm{rootscontract}(\ldots)$ function is subject to
a rootsconmode control variable, but its setting does not affect these
examples.

The boldface results in Table \ref{SystemComparison} appear to be
a consequence of happenstance more than intent, because they do not
satisfy any easily discerned goal. For example:\vspace{-0.6em}

\begin{itemize}
\item In $\boldsymbol{\frac{2\,\cdot\,98^{1/3}}{225^{1/3}}}$, both radicands
are composite with the same exponent and 225 is a perfect square,
so why not either combine the two fractional powers or simplify the
denominator to $15^{2/3}$?\vspace{-0.4em}

\item In $\boldsymbol{\frac{2}{15}\,2^{1/3}\,7^{2/3}\,15^{1/3}}$, prime
2 and composite 15 occurs to the same 1/3 power. Thus this pair could
equally well be $6^{1/3}\,5^{1/3}$ or $2^{1/3}\,10^{1/3}$. Also,
since 15 is already composite and has the same exponent as 2, why
not combine the two factors into $30^{1/3}$?\vspace{-0.4em}

\item Similar remarks apply to $\boldsymbol{\frac{2}{15}\,98^{1/3}\,15^{1/3}}$.\vspace{-0.4em}

\item In $\boldsymbol{\frac{1}{225}\,784^{1/3}\,225^{2/3}}$, $784=28^{2}$
and $225=15^{2}$, so why not express this result more comprehensibly
as $\frac{1}{15^{2}}28^{2/3}15^{4/3}\rightarrow\frac{28^{2/3}}{15^{2/3}}\rightarrow\left(\frac{28}{15}\right)^{2/3}$?\vspace{-0.4em}
\end{itemize}
\begin{defn*}
A \textsl{canonical form} for a class of expressions is one for which
all equivalent expressions in the class are represented uniquely.
\end{defn*}
As discussed in \cite{Brown,MosesSimplification,Stoutemyer10commandments},
canonical forms are unnecessarily costly and rigid for the entire
class of expressions addressed by general-purpose computer algebra
systems. However, canonical forms are acceptable and good for the
internal form systems use to represent some restricted classes of
subexpressions such as absurd numbers.

Notice that the default \textsl{Derive} result is the same for all
sixteen alternative inputs of the same absurd number, as is the default
\textsl{Mathematica} result and the Maxima radcan() result.%
\footnote{Full disclosure: We were two of the authors of \textsl{Derive}.%
} This suggests that these three columns are a consequence of transforming
absurd numbers to a canonical form. 

In contrast, none of the other columns in Table \ref{SystemComparison}
display the same result in all sixteen rows, which implies that they
are not simplified to a canonical form. Such non-canonical internal
representations might be defensible if caused by the goal of returning
the closest result to the input that satisfies one of several alternative
easily comprehended goals. However, we will explain how all of the
\textsl{inputs} already exhibit one such alternative set of goals.
Therefore maximum compliance with this goal would return the inputs
unchanged. Moreover, the dramatic transformations of most inputs throughout
Table \ref{SystemComparison} indicate that closeness to the input
was not a goal for any of these systems.

\subsection{Differences of equivalent forms of an absurd number }

It is difficult to fully simplify an expression that contains different
internal representations of the same absurd number, because syntactic
comparison is then insufficient to assess equivalence. This can lead
to a disastrously incorrect result, because if a numerator and denominator
are both equivalent to 0 but default simplification transforms only
the numerator to 0, then most default simplification will incorrectly
return 0 rather than the result of 0/0.%
\footnote{Maple and Maxima both throw an error for 0/0. Less disruptively, \textsl{Mathematica}
returns the symbol $\mathtt{Indeterminate}$ and \textsl{Derive} returns
the symbol ``?''.%
} For example, Table \ref{Differences} displays the results of default
simplification of all differences of input forms from Table \ref{SystemComparison}
for Maple and Maxima. The entry ``$0,\!\frac{0}{0}$'' indicates
that Maxima simplified the expression to $\frac{0}{0}$, but Maple
did not.

\begin{table}[H]
\caption{Default simplification of $\frac{0}{\mathrm{form}{}_{j}-\mathrm{form}_{k}}$
for Maple \& Maxima.\protect \\
\mbox{\hspace{4em}}The correct result is $\frac{0}{0}$.}
\label{Differences}

\begin{tabular}{|c|c|c|cccc|cccccccccccc}
\cline{1-3} 
\multirow{4}{*}{\begin{sideways}
$\mathrm{primal\qquad}$
\end{sideways}} & \rule[-9pt]{0pt}{26pt}$\frac{2^{4/3}\,7^{2/3}}{3^{2/3}\,5^{2/3}}$ & 1 & \textbf{$\boldsymbol{\frac{0}{0}}$} &  &  & \multicolumn{1}{c}{} &  &  &  &  &  &  &  &  &  &  &  & \tabularnewline
\cline{2-3} 
 & \rule[-9pt]{0pt}{26pt}\negthinspace{}\negthinspace{}\negthinspace{}$\frac{2}{15}{\scriptstyle \,2^{1/3}\,3^{1/3}\,5^{1/3}\,7^{2/3}}$\negthinspace{}\negthinspace{}\negthinspace{} & 2 & \textbf{$\boldsymbol{\frac{0}{0}}$} & \textbf{$\boldsymbol{\frac{0}{0}}$} &  & \multicolumn{1}{c}{} &  &  &  &  &  &  &  &  &  &  &  & \tabularnewline
\cline{2-3} 
 & \rule[-9pt]{0pt}{26pt}$\frac{14\cdot2^{1/3}\,3^{1/3}\,5^{1/3}}{15\cdot7^{1/3}}$ & 3 & \textbf{$\boldsymbol{\frac{0}{0}}$} & \textbf{$\boldsymbol{\frac{0}{0}}$} & \textbf{$\boldsymbol{\frac{0}{0}}$} & \multicolumn{1}{c}{} &  &  &  &  &  &  &  &  &  &  &  & \tabularnewline
\cline{2-3} 
 & \rule[-9pt]{0pt}{26pt}$\frac{2\cdot2^{1/3}\,7^{2/3}}{3^{2/3}\,5^{2/3}}$ & 4 & \textbf{$\boldsymbol{\frac{0}{0}}$} & \textbf{$\boldsymbol{\frac{0}{0}}$} & \textbf{$\boldsymbol{\frac{0}{0}}$} & \multicolumn{1}{c}{\textbf{$\boldsymbol{\frac{0}{0}}$}} &  &  &  &  &  &  &  &  &  &  &  & \tabularnewline
\cline{1-7} 
\multirow{6}{*}{\begin{sideways}
\textsl{\scriptsize $\begin{array}{c}
\mathrm{coprime}\:\mathrm{square\: free}\\
\mathrm{distinct\: exponents}
\end{array}\qquad\qquad$}
\end{sideways}} & \rule[-9pt]{0pt}{26pt}$\frac{2^{4/3}\,7^{2/3}}{15^{2/3}}$ & 5 & 0 & 0 & 0 & 0 & \textbf{$\boldsymbol{\frac{0}{0}}$} &  &  &  &  &  &  &  &  &  &  & \tabularnewline
\cline{2-3} 
 & \rule[-9pt]{0pt}{26pt}$\frac{2}{15}{\scriptstyle \,7^{2/3}\,30^{1/3}}$ & 6 & 0 & 0 & 0 & 0 & 0 & \textbf{$\boldsymbol{\frac{0}{0}}$} &  &  &  &  &  &  &  &  &  & \tabularnewline
\cline{2-3} 
 & \rule[-9pt]{0pt}{26pt}$\frac{14\cdot30^{1/3}}{15\cdot7^{1/3}}$ & 7 & 0 & 0 & 0 & 0 & 0 & \textbf{$\boldsymbol{\frac{0}{0}}$} & \textbf{$\boldsymbol{\frac{0}{0}}$} &  &  &  &  &  &  &  &  & \tabularnewline
\cline{2-3} 
 & \rule[-9pt]{0pt}{26pt}$\frac{2\cdot2^{1/3}\,7^{2/3}}{15^{2/3}}$ & 8 & 0 & 0 & 0 & 0 & \textbf{$\boldsymbol{\frac{0}{0}}$} & 0 & 0 & \textbf{$\boldsymbol{\frac{0}{0}}$} &  &  &  &  &  &  &  & \tabularnewline
\cline{2-3} 
 & \rule[-9pt]{0pt}{26pt}${\scriptstyle 2\cdot2^{1/3}}\left(\frac{7}{15}\right)^{2/3}$ & 9 & 0 & 0 & 0 & 0 & \textbf{$\boldsymbol{\frac{0}{0}}$} & 0 & 0 & \textbf{$\boldsymbol{\frac{0}{0}}$} & \textbf{$\boldsymbol{\frac{0}{0}}$} &  &  &  &  &  &  & \tabularnewline
\cline{2-3} 
 & \rule[-11pt]{0pt}{33pt}$\frac{14}{15}\left(\frac{30}{7}\right)^{1/3}$ & 10 & 0 & 0 & 0 & 0 & 0 & \textbf{$\boldsymbol{\frac{0}{0}}$} & \textbf{$\boldsymbol{\frac{0}{0}}$} & 0 & 0 & \textbf{$\boldsymbol{\frac{0}{0}}$} &  &  &  &  &  & \tabularnewline
\cline{1-13} 
\multirow{2}{*}{\begin{sideways}
\textsl{\scriptsize \negthinspace{}\negthinspace{}\negthinspace{}\negthinspace{}$\begin{array}{c}
\mathrm{imperfect}\\
\mathrm{powers}
\end{array}$\negthinspace{}\negthinspace{}\negthinspace{}\negthinspace{}}
\end{sideways}} & \rule[-9pt]{0pt}{26pt}$\left(\frac{28}{15}\right)^{2/3}$ & 11 & 0 & 0 & 0 & 0 & 0 & 0 & 0 & 0 & 0 & \multicolumn{1}{c|}{0} & \textbf{$\boldsymbol{\frac{0}{0}}$} &  &  &  &  & \tabularnewline
\cline{2-3} 
 & \rule[-9pt]{0pt}{26pt}$\frac{28^{2/3}}{15^{2/3}}$ & 12 & 0 & 0 & 0 & 0 & 0 & 0 & 0 & 0 & 0 & \multicolumn{1}{c|}{0} & \textbf{$\boldsymbol{\frac{0}{0}}$} & \textbf{$\boldsymbol{\frac{0}{0}}$} &  &  &  & \tabularnewline
\cline{1-15} 
\multirow{2}{*}{\begin{sideways}
\textsl{\scriptsize \negthinspace{}\negthinspace{}\negthinspace{}\negthinspace{}\negthinspace{}$\begin{array}{c}
\mathrm{reciprocal}\\
\mathrm{exponents}
\end{array}$\negthinspace{}\negthinspace{}\negthinspace{}\negthinspace{}\negthinspace{}}
\end{sideways}} & \rule[-9pt]{0pt}{26pt}$\left(\frac{784}{225}\right)^{1/3}$ & 12 & 0 & 0 & 0 & 0 & 0 & 0 & 0 & 0 & 0 & \multicolumn{1}{c|}{0} & 0 & \multicolumn{1}{c|}{0} & \textbf{$\boldsymbol{\frac{0}{0}}$} &  &  & \tabularnewline
\cline{2-3} 
 & \rule[-9pt]{0pt}{26pt}$\frac{784^{1/3}}{225^{1/3}}$ & 14 & 0 & 0 & 0 & 0 & 0 & 0 & 0 & 0 & 0 & \multicolumn{1}{c|}{0} & 0 & \multicolumn{1}{c|}{0} & \textbf{$\boldsymbol{\frac{0}{0}}$} & \textbf{$\boldsymbol{\frac{0}{0}}$} &  & \tabularnewline
\cline{1-17} 
\multirow{2}{*}{\begin{sideways}
\textsl{\scriptsize \negthinspace{}\negthinspace{}$\begin{array}{c}
\mathrm{1\; integer}\\
\mathrm{base}
\end{array}$\negthinspace{}\negthinspace{}}
\end{sideways}} & \rule[-9pt]{0pt}{26pt}$\frac{2}{15}{\scriptstyle \,1470^{1/3}}$ & 15 & 0 & 0 & 0 & 0 & 0 & 0 & 0 & 0 & 0 & \multicolumn{1}{c|}{0} & 0 & \multicolumn{1}{c|}{0} & 0 & \multicolumn{1}{c|}{0} & \textbf{$\boldsymbol{\frac{0}{0}}$} & \tabularnewline
\cline{2-3} 
 & \rule[-9pt]{0pt}{26pt}$\frac{1}{15}{\scriptstyle \,11760^{1/3}}$ & 16 & 0 & 0 & 0 & 0 & 0 & 0 & 0 & 0 & 0 & \multicolumn{1}{c|}{0} & 0 & \multicolumn{1}{c|}{0} & 0 & \multicolumn{1}{c|}{0} & $0,\boldsymbol{\!\frac{0}{0}}$ & \textbf{$\boldsymbol{\frac{0}{0}}$}\tabularnewline
\hline 
\hline 
 & form$_{j}\:\uparrow$ & \textbf{\scriptsize $\!\!\!\uparrow\mathrm{form}\#:\!\!\!$} & \multicolumn{1}{c|}{1} & \multicolumn{1}{c|}{2} & \multicolumn{1}{c|}{3} & 4 & \multicolumn{1}{c|}{5} & \multicolumn{1}{c|}{6} & \multicolumn{1}{c|}{7} & \multicolumn{1}{c|}{8} & \multicolumn{1}{c|}{9} & \multicolumn{1}{c|}{10} & \multicolumn{1}{c|}{11} & \multicolumn{1}{c|}{12} & \multicolumn{1}{c|}{13} & \multicolumn{1}{c|}{14} & \multicolumn{1}{c|}{15} & \multicolumn{1}{c|}{16}\tabularnewline
\hline 
\multicolumn{3}{|l|}{$\uparrow$ category $\longrightarrow$} & \multicolumn{4}{c|}{primal} & \multicolumn{6}{c|}{\textsl{\scriptsize $\begin{array}{c}
\mathrm{coprime}\:\mathrm{square\: free}\\
\mathrm{distinct\: exponents}
\end{array}$}} & \multicolumn{2}{c|}{\textsl{\scriptsize \negthinspace{}\negthinspace{}\negthinspace{}\negthinspace{}$\begin{array}{c}
\mathrm{imperfect}\\
\mathrm{powers}
\end{array}$\negthinspace{}\negthinspace{}\negthinspace{}\negthinspace{}}} & \multicolumn{2}{c|}{\textsl{\scriptsize \negthinspace{}\negthinspace{}\negthinspace{}\negthinspace{}\negthinspace{}$\begin{array}{c}
\mathrm{reciprocal}\\
\mathrm{exponents}
\end{array}$\negthinspace{}\negthinspace{}\negthinspace{}\negthinspace{}\negthinspace{}}} & \multicolumn{2}{c|}{\textsl{\scriptsize \negthinspace{}\negthinspace{}$\begin{array}{c}
\mathrm{1\; integer}\\
\mathrm{base}
\end{array}$\negthinspace{}\negthinspace{}}}\tabularnewline
\hline 
\end{tabular}
\end{table}

The five categories of forms, such as the primal form, are described
in Section \ref{sec:AlternativeForms-1}. Notice that the only successes
had both forms from the same category, as evidenced by the $\frac{0}{0}$
entries being confined to blocks along the diagonal.

Maple's default simplification recognizes only 17\% of the 0 denominators
and Maxima's default simplification recognizes only 18\%. If Maple
and Maxima default simplification simplified absurd numbers to a canonical
form, then all of the entries would be $\frac{0}{0}$, as they are
for \textsl{Mathematica} and \textsl{Derive}.

The Maple $\mathrm{simplify}(\ldots)$ function \textsl{does} simplify
all of these denominators to 0 despite the fact that it does not transform
all sixteen forms to the same form. Therefore a canonical form is
not absolutely necessary for zero recognition. However zero recognition
is much more difficult to implement and often slower to execute without
a canonical form. Clearly the extra effort invested in $\mathrm{simplify}(\ldots)$
was not invested in the Maple default simplification. Unfortunately,
that can cause $\mathrm{simplify}(\ldots)$ to return incorrect results
despite its admirable sophisticated zero-recognition for absurd numbers:
\begin{equation}
\mathrm{simplify}\left(\dfrac{0}{\mathrm{form}{}_{j}-\mathrm{form}{}_{k}}\right)\label{eq:Simplify(0On0)}
\end{equation}
incorrectly returns 0 rather than the result of 0/0 for 83\% of the
differences because default simplification has already incorrectly
simplified the entire argument to 0 before $\mathrm{simplify}(\ldots)$
has a chance to simplify the denominator to 0.

For similar reasons, in Maxima
\begin{equation}
\mathrm{radcan}\left(\dfrac{0}{\mathrm{form}{}_{j}-\mathrm{form}{}_{k}}\right)\label{eq:Radcan(0On0)}
\end{equation}
incorrectly returns 0 rather than the result of 0/0 for 82\% of the
differences despite the fact that $\mathrm{radcan}\left(\ldots\right)$
produces a canonical form.

Thus as much as practical, it is important for \textsl{default} simplification
to simplify the difference between equivalent forms to 0. By far the
easiest way to implement this is to default simplify equivalent inputs
to a canonical internal form.

Some systems use a canonical form based on factoring radicands, but
do not attempt complete factorization when it becomes too costly.
For example, Table \ref{SimplificationOfADifference} compares results
for the expression
\begin{equation}
\sqrt{12345701^{2}\cdot12345709}\,-\,12345701\,\sqrt{12345709}\,,\label{eq:GcdsCanHelp}
\end{equation}
\vspace{-0.6em}
which is equivalent to 0:

\begin{table}[h]

\caption{Simplification of $\sqrt{12345701^{2}\cdot12345709}\,-\,12345701\,\sqrt{12345709}\,:$}

\label{SimplificationOfADifference}

\noindent \begin{centering}
\begin{tabular}{|c|c|c|c|c|c|c|c|c|}
\hline 
\negthinspace{}\textsl{Derive}\negthinspace{} & \multicolumn{2}{c|}{Maple} & \multicolumn{3}{c|}{\textsl{Mathematica}} & \multicolumn{3}{c|}{Maxima}\tabularnewline
\hline 
\negthinspace{}default\negthinspace{} & \negthinspace{}default\negthinspace{} & \negthinspace{}simplify\negthinspace{} & \negthinspace{}default\negthinspace{} & \negthinspace{}Simplify\negthinspace{} & \negthinspace{}FullSimplify\negthinspace{} & \negthinspace{}default\negthinspace{} & \negthinspace{}rootscontract\negthinspace{} & \negthinspace{}radcan\negthinspace{}\tabularnewline
\hline 
\hline 
0 & non-0 & 0 & non-0 & non-0 & 0 & non-0 & non-0 & non-0\tabularnewline
\hline 
\end{tabular}
\par\end{centering}

\end{table}

Therefore, $\mathrm{simplify}(\ldots)$, $\mathrm{FullSimplify}[\ldots]$,
$\mathrm{rootscontract}(\ldots)$ and $\mathrm{radcan}(\ldots)$ all
\textsl{incorrectly} give 0 for the argument\vspace{-0.6em}
\[
\dfrac{0}{\sqrt{12345701^{2}\cdot12345709}\,-\,12345701\,\sqrt{12345709}}
\]
because their system's default simplification did not simplify to
0 a subexpression that is equivalent to 0. The damage was done \textsl{before}
entering these four \textsl{extra simplification} functions.

Even with some non-canonical internal forms, gcds can be used to determine
when two surds are rational multiples of each other, then combine
them. For expression (\ref{eq:GcdsCanHelp}) with the left radicand
expanded, the gcd of the two radicands is 12345709, so the expression
is equivalent to
\begin{equation}
\sqrt{15241557021803401}\cdot\sqrt{12345709}\,-\,12345701\,\sqrt{12345709},\label{eq:UseGCDStep2}
\end{equation}
then $\gcd(15241557021803401,12345701)\rightarrow12345701$ and
\begin{eqnarray}
15241557021803401/12345701 & \rightarrow & 12345701\label{eq:UseGCDStep3}
\end{eqnarray}
giving
\begin{eqnarray}
\sqrt{12345701^{2}}\cdot\sqrt{12345709}\,-\,12345701\,\sqrt{12345709} & \rightarrow & 0\label{eq:UseGCDSTep4}
\end{eqnarray}
 As indicated by the great variety of inputs and results in Table
\ref{SystemComparison}, there are a bewildering number of ways absurd
numbers can be internally stored and displayed. However, the sixteen
canonical inputs can be organized into a spectrum of categories discussed
in Section \ref{sec:AlternativeForms-1}, which compares their algorithms,
advantages, and disadvantages, with conclusions in Section \ref{sec:Conclusion}.

This article is complementary to \cite{StoutemyerProductsOfPowers},
which instead addresses products of fractional powers of \textsl{rational
powers of non-numeric expressions}. The difficulties there are different,
entailing the need to be correct even when a subsequent substitution
makes a denominator 0 or a radicand non-positive.

\section{A spectrum of representation categories\label{sec:AlternativeForms-1}}

For efficiency and ease of implementation, most computer algebra systems
use an internal representation during simplification that does not
completely correspond to displayed results. For example, often subtraction
is represented using multiplication by a negative numeric coefficient.

We can do this for absurd numbers too: We can choose an internal representation
that is easy to implement and/or fast to simplify; but for each example
display the most concise of some alternative forms. As proposed in
\cite{StoutemyerInterface,UsefulNumbers}, we could also have a transformation
wizard that opens an Alternative Transformation dialog box for a highlighted
subexpression. For example if the highlighted subexpression is $\frac{1}{225}784^{1/3}\,225^{2/3}$,
then the dialog box might be

\noindent \begin{center}
\includegraphics{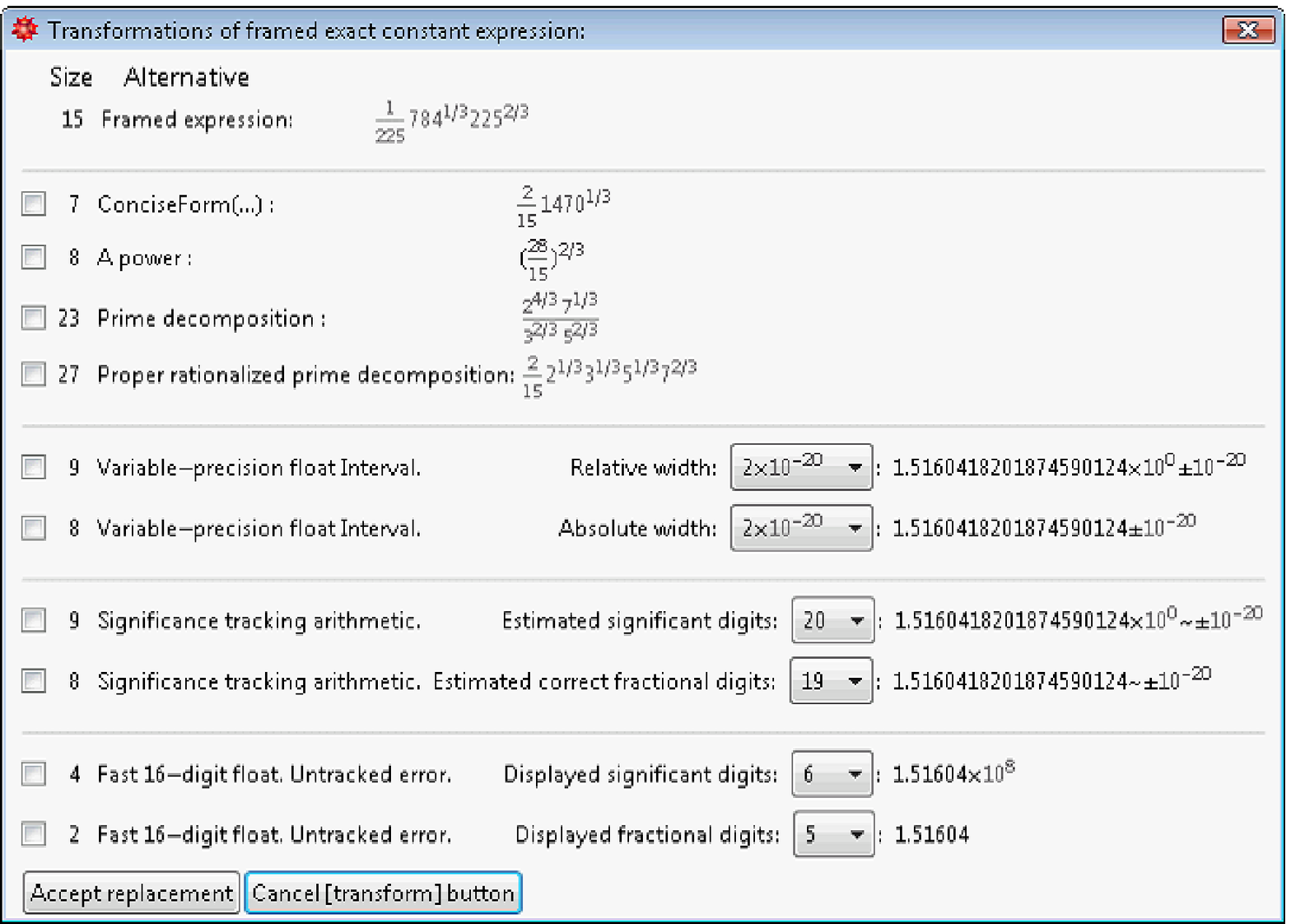}
\par\end{center}

\noindent This is a mock up created with the \textsl{Mathematica}
$\mathrm{CreateDialog}[\ldots]$ function. The size column would be
some easily-computed measure that correlates positively reasonably
well with the area used to display the alternative.

The integer or rational number bases of fractional powers can be treated
similar to variables in a data structure, but perhaps with additional
rules for when an exponent becomes integer. Therefore we can represent
expressions having more than one such prime base in either distributed
or recursive form. For example,
\[
\frac{7}{5}2^{2/3}\,3^{4/5}\,5^{1/2}-\frac{8}{5}3^{4/5}\,5^{1/2}+3^{2/3}\,5^{1/2}+3^{1/4}+6
\]
\emph{versus\vspace{-0.7em}
}

\[
\left(\left(\frac{7}{5}2^{2/3}-\frac{8}{5}\right)3^{4/5}+3^{2/3}\right)5^{1/2}+3^{1/4}+6.
\]
For brevity the forms and algorithms discussed in this section assume
only distributed form. However they can be adapted to recursive form,
which has certain advantages such as often being more concise. For
all forms we assume some canonical ordering of the factors in that
form, such as in order of increasing base magnitude.

This section discusses the sixteen numbered input forms in Table \ref{SystemComparison},
in that order. These examples are partitioned into five categories
depending on the properties of the radicands and their exponents.
To help show the relationships of the alternatives between and within
each category, Table \ref{Flo:AlternateAbSurds} shows the same input
forms in the same order, with descriptive phrases describing the properties
of the radicands, their exponents and restrictions on any rational
coefficient. This list of alternative forms is incomplete. However,
it fits on one page; it collectively suggests additional forms; and
we hope that it includes all of the most useful forms in general \textbf{--}
not only for this example.

\begin{table}[!h]
\caption{Some alternative forms for the same absurd number:}
\label{Flo:AlternateAbSurds}

\begin{tabular}{|c|c|c|c|c|}
\hline 
{\small \negthinspace{}\negthinspace{}\#\negthinspace{}\negthinspace{}} & \textbf{Ratio form} & \textbf{Product form} & \negthinspace{}\negthinspace{}\textbf{\small Coef}\negthinspace{}\negthinspace{} & \textbf{Name of form}\tabularnewline
\hline 
\hline 
\multicolumn{4}{|c|}{} & \textbf{(Bases are primes:)}\tabularnewline
\hline 
{\small 1} & \rule[-11pt]{0pt}{33pt}$\dfrac{2^{4/3}\,7^{2/3}}{3^{2/3}\,5^{2/3}}$ & $2^{4/3}\,3^{-2/3}\,5^{-2/3}\,7^{2/3}$ & $\pm1$ & pure primal\tabularnewline
\hline 
{\small 2} & \rule[-11pt]{0pt}{33pt}$\!\!\dfrac{2\,2^{1/3}\,3^{1/3}\,5^{1/3}\,7^{2/3}}{15}\!\!$ & $\frac{2}{15}\,2^{1/3}\,3^{1/3}\,5^{1/3}\,7^{2/3}$ & $\mathbb{Q}$ & proper-exponent primal\tabularnewline
\hline 
{\small 3} & \rule[-11pt]{0pt}{34pt}$\!\!\dfrac{14\!\cdot\!2^{1/3}\,3^{1/3}\,5^{1/3}}{15\!\cdot\!7^{1/3}}\negthinspace\!$ & $\frac{14}{15}\,2^{1/3}\,3^{1/3}\,5^{1/3}\,7^{-1/3}$ & $\mathbb{Q}$ & \negthinspace{}tight balanced-exponent primal\tabularnewline
\hline 
{\small 4} & \rule[-11pt]{0pt}{33pt}$\!\dfrac{2\!\cdot\!2^{1/3}\,7^{2/3}}{3^{2/3}\,5^{2/3}}\negthinspace$ & $\!\!2\!\cdot\!2^{1/3}\,3^{-2/3}\,5^{-2/3}\,7^{2/3}\!\!$ & $\mathbb{Q}$ & loose balanced-exponent primal\tabularnewline
\hline 
\multicolumn{4}{|c|}{} & \textbf{(Bases are coprime and square-free}:)\tabularnewline
\hline 
{\small 5} & \rule[-11pt]{0pt}{33pt}$\dfrac{2^{4/3}\,7^{2/3}}{15^{2/3}}$ & $2^{4/3}\,7^{2/3}\,15^{-2/3}$ & $\pm1$ & {\small $\begin{array}{c}
\mathrm{coprime\; square\; free\; integer\; bases,}\\
\mathrm{distinct\; exponents}
\end{array}$}\tabularnewline
\hline 
{\small 6} & \rule[-11pt]{0pt}{33pt}$\dfrac{2\,7^{2/3}\,30^{1/3}}{15}$ & $\frac{2}{15}\,7^{2/3}\,30^{1/3}$ & $\mathbb{Q}$ & {\small $\begin{array}{c}
\mathrm{coprime\; square\; free\; integer\; bases,}\\
\mathrm{distinct\; proper\; exponents}
\end{array}$}\tabularnewline
\hline 
7 & \rule[-11pt]{0pt}{33pt}$\dfrac{14\,30^{1/3}}{15\,7^{1/3}}$ & $\frac{14}{15}\,7^{-1/3}\,30^{1/3}$ & $\mathbb{Q}$ & {\small $\begin{array}{c}
\mathrm{coprime\; square\; free\; integer\; bases,}\\
\mathrm{distinct\; tight\; balanced\; exponents}
\end{array}$}\tabularnewline
\hline 
{\small 8} & \rule[-11pt]{0pt}{33pt}$\!\dfrac{2\!\cdot\!2^{1/3}\,7^{2/3}}{15^{2/3}}\!$ & $2\cdot2^{1/3}\,7^{2/3}\,15^{-2/3}$ & $\mathbb{Q}$ & {\small $\begin{array}{c}
\mathrm{coprime\; square\; free\; integer\; bases,}\\
\mathrm{distinct\; loose\; balanced\; exponents}
\end{array}$}\tabularnewline
\hline 
{\small \negthinspace{}\negthinspace{}9\negthinspace{}\negthinspace{}} & \rule[-8pt]{0pt}{24pt}$2\!\cdot\!2^{1/3}\left(\frac{7}{15}\right)^{2/3}$ & $2\!\cdot\!2^{1/3}\left(\frac{7}{15}\right)^{2/3}$ & $\mathbb{Q}$ & {\small $\begin{array}{c}
\mathrm{coprime\; square\; free\; rational\; bases,}\\
\mathrm{distinct\; proper\; exponents}
\end{array}$}\tabularnewline
\hline 
{\small \negthinspace{}\negthinspace{}10\negthinspace{}\negthinspace{}} & \rule[-11pt]{0pt}{33pt}$\frac{14}{15}\left(\frac{30}{7}\right)^{1/3}$ & $\frac{14}{15}\left(\frac{30}{7}\right)^{1/3}$ & $\mathbb{Q}$ & {\small $\begin{array}{c}
\mathrm{coprime\; square\; free\; rational\; bases,}\\
\mathrm{distinct\; tight\; balanced\; exponents}
\end{array}$}\tabularnewline
\hline 
\multicolumn{4}{|c|}{} & \textbf{(Bases are imperfect powers:)}\tabularnewline
\hline 
{\small \negthinspace{}\negthinspace{}11\negthinspace{}\negthinspace{}} & \rule[-8pt]{0pt}{24pt}$\left(\frac{28}{15}\right)^{2/3}$ & $\left(\frac{28}{15}\right)^{2/3}$ & $\pm1$ & {\small $\begin{array}{c}
\mathrm{single\; rational\; imperfect\; power\; base,}\\
\mathrm{positive\; exponent}
\end{array}$}\tabularnewline
\hline 
{\small \negthinspace{}\negthinspace{}12\negthinspace{}\negthinspace{}} & \rule[-11pt]{0pt}{33pt}$\dfrac{28^{2/3}}{15^{2/3}}$ & $15^{-2/3}\,28^{2/3}$ & $\pm1$ & {\small $\begin{array}{c}
\mathrm{ratio\; of\; two\; imperfect\; power\; integer\; bases,}\\
\mathrm{positive\; exponents}
\end{array}$}\tabularnewline
\hline 
\multicolumn{4}{|c|}{} & \textbf{(Maximal reciprocal exponents}:)\tabularnewline
\hline 
{\small \negthinspace{}\negthinspace{}13\negthinspace{}\negthinspace{}} & \rule[-8pt]{0pt}{24pt}$\left(\frac{784}{225}\right)^{1/3}$ & $\left(\frac{784}{225}\right)^{1/3}$ & $\pm1$ & {\small $\begin{array}{c}
\mathrm{single\; rational\; base,}\\
\mathrm{maximal\; positive\; reciprocal\; exponent}
\end{array}$}\tabularnewline
\hline 
{\small \negthinspace{}\negthinspace{}14\negthinspace{}\negthinspace{}} & \rule[-11pt]{0pt}{33pt}$\dfrac{784^{1/3}}{225^{1/3}}$ & $225^{-1/3}\,784^{1/3}$ & $\pm1$ & {\small $\begin{array}{c}
\mathrm{ratio\; of\; two\; integer\; bases,}\\
\mathrm{maximal\; positive\; reciprocal\; exponents}
\end{array}$}\tabularnewline
\hline 
\multicolumn{4}{|c|}{} & \textbf{(One integer base:)}\tabularnewline
\hline 
{\small \negthinspace{}\negthinspace{}15\negthinspace{}\negthinspace{}} & \rule[-11pt]{0pt}{33pt}$\dfrac{2\,1470^{1/3}}{15}$ & $\frac{2}{15}\,1470^{1/3}$ & $\mathbb{Q}$ & {\small $\begin{array}{c}
\mathrm{single\; minimal\; integer\; base,}\\
\mathrm{proper\; exponent}
\end{array}$}\tabularnewline
\hline 
\negthinspace{}\negthinspace{}16\negthinspace{}\negthinspace{} & \rule[-11pt]{0pt}{33pt}$\dfrac{11760^{1/3}}{15}$ & $\frac{1}{15}\,11760^{1/3}$ & $\mathbb{Q}$ & {\small $\begin{array}{c}
\mathrm{single\; integer\; imperfect\; power\; base,}\\
\mathrm{proper\; exponent}
\end{array}$}\tabularnewline
\hline 
\end{tabular}
\end{table}

\subsection{Primal forms}

There are several canonical forms for absurd numbers based on prime
factorization.

\subsubsection{Pure primal form\emph{\vspace{0.3em}
}}
\begin{defn*}
\textsl{Pure primal} form is 0, 1, or a product of one or more distinct
primes raised to nonzero reduced rational powers, with the factors
in increasing order of their prime bases, or this form preceded by
a minus sign.\end{defn*}
\begin{rem*}
For example, input \#1 in Table \ref{Flo:AlternateAbSurds} has this
form.\end{rem*}
\begin{prop}
Pure primal form is canonical for absurd numbers.\end{prop}
\begin{proof}
Adapt almost any proof of the fundamental theorem of arithmetic from
positive integer to reduced rational exponents, and employ the chosen
canonical ordering of factors.
\end{proof}
The algorithm for \textsl{multiplication} of two absurd numbers represented
using pure primal forms is obvious, easy to implement, and fast \textbf{--}
as is raising a pure primal form to a rational power. Unfortunately
for computer algebra implementers:
\begin{enumerate}
\item Conversion of composite radicands to primes requires integer factoring.
\item Factorization of large integers can be prohibitively time consuming.
\item Large integers occur rather often within computer algebra \textbf{--}
even when the input and final result do not contain large integers.
\end{enumerate}
\textsl{Addition} of pure primal forms is also quite slow because
there is very little opportunity for making results more concise by
merely combining \textsl{syntactically} similar terms: The only syntactically
similar terms are ones that are identical or differ only in their
signs, such as
\[
2^{3/2}\,5^{-7/3}+\left(-2^{3/2}\,5^{-7/3}\right)\;\rightarrow\;(1-1)2^{3/2}\,5^{-7/3}\;\rightarrow\;0\,.
\]
Thus to simplify $2^{5/2}\,5^{-5/3}-2^{1/2}\,5^{-2/3}$ to this canonical
form we must recognize that the differences in their corresponding
exponents are all integers, then \textsl{temporarily} use a form that
makes them syntactically similar such as the proper-exponent primal
form discussed in the next subsection. Then in general we must factor
the resulting rational coefficient to convert the result to pure primal
form:
\begin{multline*}
2^{5/2}\,5^{-5/3}-2^{1/2}\,5^{-2/3}\;\rightarrow\;\frac{4}{25}\left(2^{1/2}\,5^{1/3}\right)-\frac{1}{5}\left(2^{1/2}\,5^{1/3}\right)\;\rightarrow\;-\frac{1}{25}\left(2^{1/2}\,5^{1/3}\right)\;\rightarrow\;\\
-5^{-2}\left(2^{1/2}\,5^{1/3}\right)\;\rightarrow\;-2^{1/2}\,5^{-5/3}.
\end{multline*}
If the differences in the exponents are not all integer, then the
two numbers are not commensurate and their sum or difference must
be represented as a more general expression than a single absurd number.

To incur the cost of integer factorization not only initially but
also after most additions of similar terms makes pure primal form
a costly internal form, but it can be useful as an optional \textsl{result}
form. Most computer algebra systems contain a rational-number factorization
function that returns the pure primal form.

\subsubsection{Proper-exponent primal form\label{sub:Proper-primal-forms}\emph{\vspace{0.3em}
}}
\begin{defn*}
\textsl{Proper-exponent primal} form is a rational number times a
pure primal form in which all of the exponents are in the interval
$(0,1)$.\end{defn*}
\begin{rem*}
For example, input \#2 in Table \ref{Flo:AlternateAbSurds} is proper-exponent
primal form.
\end{rem*}
We were taught in beginning algebra to simplify a fractional power
of a positive rational number $r^{\alpha}$ by converting it to this
form. The algorithm \textsl{can} be expressed as
\begin{enumerate}
\item Factor $r$.
\item Represent the factored $r$ as a product containing negative powers
rather than as a ratio.
\item Distribute $\alpha$ over the factors.
\item Extract and multiply together the rational numbers corresponding to
the \textsl{floor} of each resulting exponent.
\end{enumerate}
Using the floor automatically rationalizes the denominator.
\begin{prop}
Proper-exponent primal form is canonical.\end{prop}
\begin{proof}
Step 3 above gives the canonical pure primal form. The floor function
is defined and single valued for all reals. Thus the extracted rational
parts, their product, and any residual fractional powers are unique.
Conversely, if we start with the proper-exponent primal form, factor
the rational part, combine similar primes and order the bases, then
the result is unique.
\end{proof}
The algorithms for multiplication of two proper-exponent primal forms
and for raising one to a rational power are nearly as obvious, easy
to implement, and fast as for pure primal form. Moreover addition
of proper-exponent primal forms is much easier and more efficient
than for pure primal forms: Sums of proper-exponent primal forms can
be collected to make another such form if and only if the irrational
factors are \textsl{identical} which is fast to check. Moreover, when
the irrational factors are identical, there is no need to factor the
resulting rational coefficient.

Unfortunately, integer factorization is still generally needed to
transform a fractional power of a positive rational number to proper-exponent
primal form.

Pure primal form often requires less display area than other primal
forms because there is no rational factor formed from expanding a
product of integer powers of the prime bases having exponents not
in a designated interval. However, users often feel that extracting
that rational factor makes the result ``simpler''. Students are
also taught to rationalize denominators. Therefore, display of proper-exponent
primal form complies with users' comfort zones.

\begin{flushright}
``... \textsl{the customer is always right}.''\\
\textbf{--} Marshall Field
\par\end{flushright}

\subsubsection{Tight balanced-exponent primal form}

Any fixed near-unit-width exponent interval can be used instead of
$(0,1)$ for primal and other categories of forms. The nearly balanced
interval $(-1/2,1/2]$ has some appeal because the magnitude of the
exponents never exceeds 1/2. 
\begin{defn*}
\textsl{Tight balanced-exponent primal} form is a rational number
times a pure primal form in which all of the exponents are in the
interval $(-1/2,1/2]$.\end{defn*}
\begin{rem*}
For example, input \#3 in Table \ref{Flo:AlternateAbSurds} has this
form. As another example with input $2^{2/3}$ we rationalize the
\textsl{numerator} to return $2/2^{1/3}$.
\end{rem*}

\subsubsection{Loose balanced-exponent primal form\emph{\vspace{0.3em}
}}
\begin{defn*}
\textsl{Loose balanced-exponent primal} form is a rational coefficient
times a pure primal form in which all of the exponents are in the
interval $(-1,1)$ , none of the numerator radicands divide the denominator
of the rational coefficient, and none of the denominator radicands
divide the numerator of the rational coefficient.\end{defn*}
\begin{rem*}
For example, input \#4 in Table \ref{Flo:AlternateAbSurds} has this
form.
\end{rem*}
This form can be derived from pure primal form by separately making
the numerator and denominator exponents proper, and \textsl{not} rationalizing
any denominators or numerators. For example, $\sqrt{2}$, $1/\sqrt{2}$,
and $5\times7^{2/3}/(2\times3^{4/5})$ have this form, but $\sqrt{2}/6$
does not because $\gcd(2,6)\neq1$. Loose balanced-exponent form is
often more concise than proper-exponent or tight-balanced primal form
because rationalizing denominators or numerators often increases bulk.
For example, compare $1/\sqrt{1234567891}$ with $\sqrt{1234567891}/1234567891$
and compare $\sqrt{9876543211}$ with
\[
9876543211/\sqrt{9876543211}.
\]
 Such rationalizations are also inconsistent with customary simplification
of \textsl{non-numeric} radicands: Most people prefer $1/u^{1/3}$
or $u^{-1/3}$ to the unreduced $u^{2/3}/u$, and avoiding unnecessary
form changes upon substitution of numbers is a virtue. We can of course
use a product containing negative exponents rather than a ratio for
the internal form.

However, addition is harder with loose balanced-exponent primal form
than with proper-exponent primal form, because mere syntactic comparison
does not reveal all commensurate absurd numbers. For example, $2^{1/2}$
and $2^{-1/2}$ are not syntactically similar, but $2^{1/2}+2^{-1/2}\rightarrow2^{1/2}+2^{1/2}/2\rightarrow(3/2)2^{1/2}\rightarrow3\times2^{-1/2}$
. Therefore we do not recommend this as an \textsl{internal} form,
but it is a good display form.

\subsection{Coprime square-free distinct-exponent forms}

Gcd calculations are sufficient to make all fractional power bases
in a result mutually coprime. For example, $\gcd(30,42)\rightarrow6$,
so
\[
30^{1/2}42^{1/3}\;\rightarrow\;\left(5\cdot6\right)^{1/2}\,\left(6\cdot7\right)^{1/3}\;\rightarrow\;5^{1/2}\,6^{1/2}\,6^{1/3}\,7^{1/3}\;\rightarrow\;5^{1/2}\,6^{5/6}\,7^{1/3}.
\]

This \textsl{particular} result is canonical, but relative primality
alone is not sufficient to guarantee canonicality. For example, the
bases in the equivalent forms $5^{1/2}\,6^{1/2}\,7^{1/3}$ and $3^{1/2}\,7^{1/3}\,10^{1/2}$
are coprime. Combining factors having the same exponents to make all
of the exponents distinct makes this example canonical: We combine
$5^{1/3}$ and $6^{1/3}$ or combine $3^{1/3}$ and $10^{1/3}$ giving
$7^{1/3}\,30^{1/2}$ either way.

However, 24 and 5 are coprime in $24^{1/2}\,5^{1/3}$ with distinct
exponents, as are 2, 3 and 5 in $2^{3/2}\cdot3^{1/2}\,5^{1/3}$, but
both products are equivalent. So we need an additional criterion for
canonicality:
\begin{defn*}
An integer > 1 is \textsl{square free} if none of its prime factors
occurs more than once.\end{defn*}
\begin{rem*}
For example, 6 is square free, but $24=2^{3}\,3$ is not.\end{rem*}
\begin{defn*}
\textsl{coprime square-free integer bases distinct-exponent} form
is a product of coprime positive square-free integers raised to distinct
rational exponents, or {\small $-1$} times that, or 0.\end{defn*}
\begin{rem*}
For example, input \#5 in Table \ref{Flo:AlternateAbSurds} has this
form.
\end{rem*}
This form can be computed from pure primal form as follows: For each
distinct exponent, combine all of the factors having that exponent,
raising the product of their primes to their shared exponent. There
is clearly only one way to do this, the resulting bases are square-free
because they are a product of distinct primes, and the resulting bases
are coprime because each prime occurs in only one of the bases.

Conversely, to compute the pure primal form from this form, factor
each base then distribute the distinct exponent of that square-free
base over the resulting product of primes. Each distinct prime can
occur in only one of the coprime factors, so the distinct exponent
for each base will be the final exponent of all the primes in that
base. This result is clearly unique when the bases are ordered canonically,
so the coprime square-free integer bases distinct-exponent form is
canonical.

When multiplying two such forms, if a base $b$ in one form is not
identical to a base in the other form, then it is important to compute
the gcd of $b$ with the bases in the other form to check for coprimeness
and act appropriately if any of these gcds is not 1. It is also important
to check for identical exponents as well as identical bases. Therefore
multiplication is slower and not as easy to implement as for primal
forms.

There are also proper, tight balanced, and loose balanced variants
analogous to those based on prime radicands. For example,
\begin{itemize}
\item Input \#6 is\textsl{ coprime square-free integer bases distinct proper-exponent}
form. Two such forms are commensurate for addition if and only if
their irrational parts are identical, making this variant the best
choice of internal form for this class. 
\item Input \#7 is \textsl{coprime square-free integer bases distinct tight
balanced-exponent} form.
\item Input \#8 is \textsl{coprime square-free integer bases distinct loose
balanced-exponent} form.
\end{itemize}
As with primal forms, the proper variant is most efficient for adding
absurd numbers because syntactic comparison is sufficient to decide
similarity.

We can also combine factors whose exponents differ only in sign, for
further sharing of common exponents, giving forms that are more concise
and faster for subsequent floating-point approximation:
\begin{defn*}
\textsl{coprime square-free rational bases distinct proper-exponent
}form is a rational number times a product of coprime square-free
positive rational numbers raised to distinct proper exponents, or
{\small $-1$} times that, or 0.\end{defn*}
\begin{rem*}
For example, input \#9 is that form, obtained by combining a numerator
and denominator factor of input \#5.\end{rem*}
\begin{defn*}
\textsl{coprime square-free rational bases distinct tight balanced
exponent }form is a product of coprime square-free positive rational
numbers raised to distinct exponents in the interval $(-1/2,1/2]$,
or {\small $-1$} times that, or 0.\end{defn*}
\begin{rem*}
For example, input \#10 is that form, obtained by combining a numerator
and denominator factor of input \#7.
\end{rem*}
Unfortunately there is no known way to square-free factor an integer
faster than by factoring it then combining bases having identical
exponents. For example, to square-free factor 2910600, we can factor
it into $2^{3}\,3^{3}\,5^{2}\,7^{2}11$, then combine factors having
the same exponent to produce $6^{3}\,35^{2}\,11$ in which 6, 35 and
11 are square free. However, if we are incurring the cost of integer
factorization anyway, it is simpler and probably faster to use full
factorization for the internal form. Consequently although these coprime
square-free distinct-exponent forms are often the most concise \textsl{display}
forms, we do not recommend them as an \textsl{internal} form.

\subsection{Forms based on perfect power computation}

Our next family of canonical forms uses perfect-power factorization
to avoid the cost of integer factorization.
\begin{defn*}
For integers $k>1$, $m>1$ and $n>1$, $n$ is a $k$\textsuperscript{th}
\textsl{perfect power} of $m$ if and only if $m^{k}=n$.
\end{defn*}
We are interested in determining the maximum $\hat{k}$ and minimum
integer $\check{m}$ for which $\check{m}^{\hat{k}}=n$.

Perfect powers can be determined from a prime factorization, but there
is a much faster way: As described by Fitch \cite{FitchNewtonNthRootOfInteger},
given integers $n>1$ and $k>1$, Newton's method with the help of
the floor function can be used to quickly compute an \textsl{exact}
$k$\textsuperscript{th} root of $n$ or determine that one does
not exist. If $n$ has $b$ bits, then starting with a guess that
has $\left\lceil b/k\right\rceil $ bits, the number of correct bits
of the result doubles from 1 or more with each iteration, which is
fast. Since we want the largest such $k$, we can try successive primes
starting with 2, repeating each prime until it no longer works. Each
success substantially reduces $m$ from its initial value of $n$,
reducing the work for subsequent trials. We can stop when for the
next prime $p$ and the current value of $m$, $2^{p}>m$. Here are
some extreme examples for large \textsl{$n$}, ordered from least
to most applications of Newton's method:
\begin{enumerate}
\item For expanded $n=6^{2\cdot3\cdot5\cdot7}$, it requires one successful
and one unsuccessful application of Newton's method with each of the
4 successive primes 2, 3, 5 and 7 to determine that the 164 digit
$n=6^{210}$.
\item For expanded $n=2^{2^{9}}$ it requires 9 successful applications
of Newton's method with the first-tried prime 2 to determine that
the 155 digit $n=2^{512}$.
\item For expanded $n=m^{2}$ with $m$ being the largest prime less than
$2^{256}$, it requires 1 successful application of Newton's method
followed by 53 unsuccessful applications to determine that the 154
digit $n=m^{2}$.
\item For expanded $n=2^{509}$, it requires 96 unsuccessful applications
with successive primes followed by one successful application with
the prime 509 to determine that the 154 digit $n=2^{509}$.
\item For $n$ being the largest prime less than $2^{512}$, it requires
97 unsuccessful applications with successive primes to determine that
the 154 digit $n$ is not a perfect power.%
\footnote{It might be worth using a primality test in perfect power factorization.%
}
\end{enumerate}
For both the mean case and worst case this method is much faster than
factoring large integers.
\begin{defn*}
A positive rational number is a \textsl{perfect} $k$\textsuperscript{th}
\textsl{power} if it is a perfect $k$\textsuperscript{th} power
of an integer, or the reciprocal of such a perfect power, or its numerator
is a perfect $j$\textsuperscript{th} power, its denominator is a
perfect $\ell$\textsuperscript{th} power, and $k$ evenly divides
$\gcd(j,\ell)$.\emph{\vspace{-0.8em}
}
\end{defn*}
\,
\begin{defn*}
A positive rational number is an \textsl{imperfect power} if it is
not a perfect power.
\end{defn*}
We are interested in determining the maximum $\hat{k}$ and minimum
rational number $\check{r}$ for which $\check{r}^{\hat{k}}=r>1$.
For a reduced positive fraction $r$ that is neither an integer nor
a reciprocal, we can compute the $\hat{k}_{1}$ for whichever of the
numerator and denominator is smaller, then if $\hat{k}_{1}>1$, restrict
the choice of primes for applying Newton's method to the other part
to primes that exactly divide $\hat{k}_{1}$.

\subsubsection{A positive rational power of a positive rational number that is an
imperfect power\emph{\vspace{0.3em}
}}
\begin{defn*}
\textsl{Single rational imperfect power base positive exponent} form
is a positive rational power of a positive rational number that is
an imperfect power, or {\small $-1$} times that, or 0.\end{defn*}
\begin{rem*}
For example, input \#11 in Table \ref{Flo:AlternateAbSurds} has this
form.
\end{rem*}
To transform a positive rational power $\alpha$ of a positive rational
number $r$ to this form, maximally perfect-power factor $r\rightarrow\check{r}^{\hat{k}}$,
then return $\check{r}^{\hat{k}\alpha}$.
\begin{prop}
\textsl{Single rational imperfect power base positive exponent form}
is canonical.\end{prop}
\begin{proof}
To convert pure primal representation $P=p_{1}^{\alpha_{1}}p_{2}^{\alpha_{2}}\cdots$
to this form, let
\begin{eqnarray*}
\gamma & \leftarrow & \gcd(\alpha_{1},\alpha_{2},\ldots),
\end{eqnarray*}
which is positive, then let $n_{1}\leftarrow\alpha_{1}/\gamma$, $n_{2}\leftarrow\alpha_{2}/\gamma$,
etc. All of the $n_{j}$ are coprime \textsl{integers}. Consequently
$P$ can be represented as $r^{\gamma}$ where $r$ is the expanded
rational number $p_{1}^{n_{1}}p_{2}^{n_{2}}\cdots$. Then use Newton's
method to express $r$ as an imperfect power base raised to a positive
exponent: $r\rightarrow\check{r}^{\hat{k}}$ giving $P\rightarrow\check{r}^{\hat{k}\gamma}$.
The pure primal representation together with $\gamma$, $n_{1}$,
$n_{2}$, $r$, $\check{r}$ and $\hat{k}$ are all unique, therefore
this single power form for $P$ is unique. Now consider the other
direction: Factor $\check{r}$, distribute $\hat{k}\gamma$, then
sort the factors into canonical order, giving the canonical pure primal
form.
\end{proof}
To \textsl{multiply} two such forms $u_{1}r_{1}^{\alpha_{1}}$ and
$u_{2}r_{2}^{\alpha_{2}}$ with $u_{1},u_{2}\in\{1,-1\}$:
\begin{enumerate}
\item Let $\gamma\leftarrow\gcd(\alpha_{1},\alpha_{2})$, $n_{1}\leftarrow\alpha_{1}/\gamma$,
$n_{2}\leftarrow\alpha_{2}/\gamma$, making $n_{1}$ and $n_{2}$
integer
\item Use Newton's method to compute $r_{1}^{n_{1}}r_{2}^{n_{2}}\rightarrow\check{r}^{\hat{k}}$,
then return $u_{1}u_{2}\left(\check{r}\right)^{\hat{k}\gamma}$.
\end{enumerate}
To \textsl{add} two such forms:
\begin{enumerate}
\item Let $m_{1}\leftarrow\left\lfloor \alpha_{1}\right\rfloor $, $m_{2}\leftarrow\left\lfloor \alpha_{2}\right\rfloor $,
$\beta_{1}\leftarrow\alpha_{1}-m_{1}$, $\beta_{2}\leftarrow\alpha_{2}-m_{2}$.
\item If $\beta_{1}\neq\beta_{2}$, then the ratio of the two inputs is
irrational, so their sum cannot be represented as a single absurd
number.
\item Otherwise, let $g\leftarrow\gcd(r_{1},r_{2})$, $\bar{r}_{1}\leftarrow r_{1}/g$,
$\bar{r}_{2}\leftarrow r_{2}/g$, $n\leftarrow\mathrm{numerator}(\beta_{1})$
, $d\leftarrow\mathrm{denominator}(\beta_{1})$.
\item If $\bar{r}_{1}\not\rightarrow\breve{r}_{1}^{d}$ or $\bar{r}_{2}\not\rightarrow\breve{r}_{2}^{d}$
then the sum cannot be represented as a single absurd number.
\item Otherwise let $\rho\leftarrow u_{1}r_{1}^{m_{1}}\breve{r}_{1}^{n}+u_{2}r_{2}^{m_{2}}\breve{r}_{2}^{n}$.
\item If $\rho=0$ then return 0.
\item Use the multiplication algorithm to return the result of $\rho g^{\gamma}$
as a positive rational power of a positive rational number that is
an imperfect power \textbf{--} or as {\small $-1$} times that.
\end{enumerate}
This is a canonical form that avoids the cost of integer factorization!

However, for strict consistency, even a rational number might have
to be represented as a perfect power such as $256/81\rightarrow(4/3)^{4}$.
The frequent use of Newton's method to maintain this would unacceptably
slow down rational arithmetic. Therefore in practice whenever a resulting
exponent is an integer, then it is better to expand the power to the
more typical representation for a rational number.

Although fractional exponents are often merely half-integer or thirds
of an integer, this form can result in large radicands. For example,
\begin{equation}
\dfrac{29^{31/10}}{2^{1/10}}\:\rightarrow\:\left(\dfrac{29^{31}}{2}\right)^{1/10}\:\rightarrow\:\left(\dfrac{2159424054808578564166497528588784562372597429}{2}\right)^{1/10}.\label{eq:NonPerfectPowRadicandGrowth}
\end{equation}

Extracting a maximal rational factor from the pure primal form made
it much faster to add absurd numbers, so it is natural to wonder if
the same is true for single rational imperfect power base positive
exponent form. For example,
\[
\left(\dfrac{576}{25}\right)^{2/3}\:\rightarrow\:\left(\dfrac{24}{5}\right)^{4/3}\:\rightarrow\:\dfrac{24}{5}\left(\dfrac{24}{5}\right)^{1/3}.
\]

Unfortunately, this is not a canonical form, because $24=2^{3}3$,
so this number also has a different representation in this form as
\begin{equation}
\dfrac{48}{5}\left(\dfrac{3}{5}\right)^{1/3},\label{eq:nonCanonicalRatTimesSinglePower}
\end{equation}
and it requires square-free integer factoring to obtain this form,
which requires integer factoring. One way to make it canonical is:
Whenever an input or a tentative result is a rational number times
a fractional power of a rational number:
\begin{enumerate}
\item Transform the product to single rational imperfect power base positive
exponent form.
\item Then use the floor function to extract a rational factor if the positive
fractional power exceeds 1, making the fractional power proper.
\end{enumerate}
For example,
\[
\dfrac{48}{5}\left(\dfrac{3}{5}\right)^{1/3}\:\rightarrow\:\left(\dfrac{48^{3}\,3}{5^{3}\,5}\right)^{1/3}\:\rightarrow\:\left(\left(\dfrac{24}{5}\right)^{4}\right)^{1/3}\:\rightarrow\:\left(\dfrac{24}{5}\right)^{4/3}\:\rightarrow\:\dfrac{24}{5}\left(\dfrac{24}{5}\right)^{1/3}.
\]
 Gcds can be used to simplify sums of absurd numbers in this form
without doing this canonicalizing reabsorption. For example, $\gcd\,(24/5,\,3/5)\rightarrow3/5$,
so
\[
\dfrac{24}{5}\left(\dfrac{24}{5}\right)^{1/3}\!\!-\dfrac{48}{5}\left(\dfrac{3}{5}\right)^{1/3}\:\rightarrow\:\,\,\dfrac{24}{5}8^{1/3}\left(\dfrac{3}{5}\right)^{1/3}\!\!-\dfrac{48}{5}\left(\dfrac{3}{5}\right)^{1/3}\:\rightarrow\:\,\,\dfrac{48}{5}\left(\dfrac{3}{5}\right)^{1/3}\!\!-\dfrac{48}{5}\left(\dfrac{3}{5}\right)^{1/3}\:\rightarrow\:\,\,0.
\]
However, the loss of canonicality for irrational absurd numbers is
still troublesome. For example, for any function $f$ \textbf{--}
including a generic one with no current definition \textbf{--} we
would like
\begin{eqnarray*}
f\left(\dfrac{24}{5}\left(\dfrac{24}{5}\right)^{1/3}\right)-f\left(\dfrac{48}{5}\left(\dfrac{3}{5}\right)^{1/3}\right) & \rightarrow & 0.
\end{eqnarray*}
However, that will not happen with this non-canonical form unless
every time a subexpression of the form $f(u)-f(v)$ is encountered
during all transformations we check to see if $u-v$ can be simplified
to 0. This is time consuming, a programming nuisance, and unlikely
to enjoy 100\% programmer compliance.

A way to partially overcome this dilemma is to use the proper variant
only temporarily during a sequence operations with irrational absurd
numbers, then represent the result of this sequence as a rational
number if it is one, or as a unit times a single power otherwise.
However, such context-dependent departure from pure locally self-contained
bottom-up simplification is extra programming work, hence an invitation
to inconsistent behavior.

In any event, it is definitely worthwhile overall to represent rational
absurd numbers as rational numbers.

\subsubsection{A ratio of positive rational powers of positive integers that are
imperfect powers\label{sub:RatioOf2NonPerfectPowers}\emph{\vspace{0.3em}
}}
\begin{defn*}
\textsl{Ratio of two imperfect power integer bases raised to positive
exponents} form is a ratio of two positive rational powers of positive
integers that are imperfect powers, or {\small $-1$} times that,
or 0.
\end{defn*}
Input \#12 in Table \ref{Flo:AlternateAbSurds} has this form. For
this example the resulting two exponents are identical, but that might
not be so. For example, this form can help reduce or avoid radicand
growth by transforming the right side of (\ref{eq:NonPerfectPowRadicandGrowth})
to the left side.

\subsubsection{Maximal positive reciprocal power of a positive rational number form\emph{\vspace{0.3em}
}}
\begin{defn*}
\textsl{Maximal positive reciprocal-exponent form} is the largest
possible positive reciprocal power of a positive rational number,
or {\small $-1$} times that, or 0.\end{defn*}
\begin{rem*}
For example, input \#13 in Table \ref{Flo:AlternateAbSurds} has this
form. As another example, $(9/4)^{1/4}$ does not have this form,
but the equivalent expression $(3/2)^{1/2}$ does. Although $8/27$
is a perfect cube, $(8/27)^{1/2}$ has this form because the exponent
in $(2/3)^{3/2}$ is not a reciprocal.
\end{rem*}
To convert a positive reduced fractional power of a positive reduced
rational number $r^{n/d}$ to this form:
\begin{enumerate}
\item Use Newton's method to find $\bar{d}$, the largest divisor of $d$
such that $r\rightarrow\bar{r}^{\bar{d}}$.
\item Expand $\bar{r}^{n}$ giving $\hat{r}$.
\item Return the form $\hat{r}^{\bar{d}/d}$. Note that this generally requires
fewer applications of Newton's method than to determine the \textsl{maximal}
perfect root of $r$.\end{enumerate}
\begin{prop}
A maximal positive reciprocal power of a positive rational number
is canonical.\end{prop}
\begin{proof}
To convert pure primal representation $P=p_{1}^{\alpha_{1}}p_{2}^{\alpha_{2}}\cdots$
to this form, let
\begin{eqnarray*}
\gamma & \leftarrow & \gcd(\alpha_{1},\alpha_{2},\ldots),
\end{eqnarray*}
which is positive, then $m_{1}\leftarrow\alpha_{1}/\gamma$, $m_{2}\leftarrow\alpha_{2}/\gamma$,
etc. (The gcd of two fractions is the gcd of their numerators divided
by the least common multiple of their denominators. Therefore all
of the multiplicities $m_{j}$ are coprime \textsl{integers}.) Consequently
$P$ can be represented as $r^{1/\mathrm{denominator}(\gamma)}$ where
$r$ is the expanded rational number $\left(p_{1}^{m_{1}}p_{2}^{m_{2}}\cdots\right)^{\mathrm{numerator}(\gamma)}$.
The pure primal representation together with $\gamma$, $m_{1}$,
$m_{2}$ and $r$ are all unique. Therefore this single power form
for $P$ is unique. Now consider the other direction: Factor $r$,
distribute $1/\mathrm{denominator}(\gamma)$, then sort the factors
into canonical order, giving the canonical pure primal form.
\end{proof}
As illustrated by comparing inputs \#11 and \#13 in Table \ref{Flo:AlternateAbSurds},
this form can be less concise than imperfect power form. However,
the arithmetic is faster:

To \textsl{multiply} two such forms $u_{1}r_{1}^{\alpha_{1}}$ and
$u_{2}r_{2}^{\alpha_{2}}$ with $u_{1},u_{2}\in\{1,-1\}$:
\begin{enumerate}
\item Let $\gamma\leftarrow\gcd(\alpha_{1},\alpha_{2})$, which is a reciprocal
because $\alpha_{1}$, and $\alpha_{2}$ are both reciprocals.
\item Let $m_{1}\leftarrow\alpha_{1}/\gamma$, $m_{2}\leftarrow\alpha_{2}/\gamma$,
making $m_{1}$ and $m_{2}$ integer.
\item Expand $u_{1}u_{2}$ giving $u$ and expand $r_{1}^{m_{1}}r_{2}^{m_{2}}$
giving $r$.
\item Convert $r^{\gamma}\rightarrow\hat{r}^{1/\check{d}}$ by the algorithm
at the beginning of this sub-subsection.
\item Return $u\hat{r}^{1/\check{d}}$.
\end{enumerate}
To \textsl{add} two such forms:
\begin{enumerate}
\item If $\alpha_{1}\neq\alpha_{2}$, then the sum cannot be represented
as a single absurd number.
\item Otherwise, let $g\leftarrow\gcd(r_{1},r_{2})$, $n_{1}\leftarrow r_{1}/g$,
$n_{2}\leftarrow r_{2}/g$, $d\leftarrow\mathrm{denominator}(\alpha_{1})$.
\item If $n_{1}\not\rightarrow\bar{n}_{1}^{d}$ or $n_{2}\not\rightarrow\bar{n}_{2}^{d}$,
then the sum cannot be represented as a single absurd number. (These
perfect root computations require only one or two applications of
Newton's method for one specific $d$, making them faster than determining
the maximal perfect powers.)
\item Otherwise let $\rho\leftarrow u_{1}\bar{n}_{1}+u_{2}\bar{n}_{2}$.
\item If $\rho=0$ then return 0.
\item Otherwise use the multiplication algorithm to transform $\rho g^{\gamma}$
into a maximal reciprocal power of a positive rational number \textbf{--}
or {\small $-1$} times that.
\end{enumerate}
This is another canonical form that avoids the cost of integer factorization,
and requires fewer applications of Newton's method than imperfect
power form. Moreover, rational numbers are a special case wherein
$\gamma$ is the reciprocal of 1. However, the radicand can become
quite large if the numerator of the given exponent is large. For example,
\[
29^{31/10}\:\rightarrow\:\left(29^{31}\right)^{1/10}\:\rightarrow\:2159424054808578564166497528588784562372597429^{1/10}.
\]

\subsubsection{A ratio of two maximal reciprocal powers of positive integers}
\begin{defn*}
\textsl{Ratio of two maximal reciprocal powers of positive integers}
form is a ratio of two maximally positive reciprocal powers of positive
integers, or {\small $-1$} times that, or 0.
\end{defn*}
Input \#14 in Table \ref{Flo:AlternateAbSurds} has this form. For
this example the resulting two exponents are identical, but that need
not be so. For example,
\[
\dfrac{2^{3/4}\,7^{1/4}}{3^{2/3}5^{1/3}}\:\rightarrow\:\dfrac{\left(8\times7\right)^{1/4}}{\left(9\times5\right)^{1/3}}\:\rightarrow\:\dfrac{56^{1/4}}{45^{1/3}}.
\]
In contrast, the radicand for unification into a single fractional
power can be significantly larger:
\[
\dfrac{56^{1/4}}{45^{1/3}}\:\rightarrow\:\dfrac{56^{3/12}}{45^{4/12}}\:\rightarrow\:\dfrac{175616^{1/12}}{4100625^{1/12}}\:\rightarrow\:\left(\dfrac{175616}{4100625}\right)^{1/12}.
\]

However, arithmetic using separate single numerator and denominator
radicals is more complicated, because we must use gcds to insure that
numerator radicands and denominator radicands are relatively prime
and contend with possibly different exponents of their gcd if it is
not 1.

\subsection{Proper power of an integer forms}

People often like to have absurd numbers displayed as a rational number
times \textsl{one} rationalized proper fractional power of the smallest
possible positive integer because:
\begin{itemize}
\item It is proper.
\item The denominator is rationalized, which students are taught to overvalue.
\item It has only one fractional power and the radicand is an integer.
\item The maximum possible amount of rational coefficient is factored out,
so the one radicand is as simple as possible for such a form.\end{itemize}
\begin{defn*}
\textsl{Single minimal integer base raised to a proper exponent} form
is a rational number or a rational number times the smallest possible
positive integer raised to an exponent in the interval $(0,1)$.

Input \#15 in Table \ref{Flo:AlternateAbSurds} has this form. Proper
exponent primal form can be converted to this form by unifying all
its fractional powers into one. For example,
\[
\dfrac{2}{15}2^{1/3}3^{1/3}5^{1/3}7^{2/3}\:\rightarrow\:\frac{2}{15}\left(2\times3\times5\times7^{2}\right)^{1/3}\:\rightarrow\:\dfrac{2}{15}1470^{1/3}.
\]

\end{defn*}
Unfortunately we do know how to guarantee this form canonical without
integer factorization. However, the following form is similar and
avoids integer factorization, but can result in larger integer radicands:
\begin{defn*}
\textsl{Single integer imperfect power base} form is derived from
the ratio of two imperfect power integer bases with positive exponents
form as follows:\end{defn*}
\begin{enumerate}
\item Extract a rational factor by independently making the numerator and
denominator exponents proper,
\item Rationalize the denominator.
\item Unify the resulting two numerator fractional powers into a single
fractional power of an integer.
\end{enumerate}
For example, starting with input \#12,
\[
\dfrac{28^{2/3}}{15^{2/3}}\;\rightarrow\;\dfrac{28^{2/3}\,15^{1/3}}{15^{2/3}\,15^{1/3}}\;\rightarrow\;\dfrac{(784\times15)^{1/3}}{15}\;\rightarrow\;\dfrac{11760^{1/3}}{15}
\]
giving input \#16 in Table \ref{Flo:AlternateAbSurds}. This is not
as nice as input \#16, but depending on taste and the application,
it is arguably nicer than inputs \#11 through \#13, which are the
only other listed ones based only on imperfect powers with no need
for integer factorization.

Although it is a good default display form for implementations that
totally avoid integer factorization, this form is not very convenient
as an internal form.

\subsection{A hybrid of proper-exponent primal and imperfect power forms}
\begin{itemize}
\item Primal internal forms entail integer factorization time that is unacceptable
to many users if the second largest prime factor is larger than about,
say, 30 digits. However, these forms are the only points of departure
for generating many display forms that users might value. Of the various
primal forms, arithmetic is the fastest and easiest to implement for
proper-exponent primal form, with tight balanced primal form being
a close second.
\item coprime square-free distinct exponent internal forms also occasionally
entail unacceptable integer factorization time, but they tend to require
less display space than primal forms. However, coprime forms are a
helpful point of departure for fewer display forms, and the arithmetic
is somewhat slower and harder to implement.
\item Internal forms based on imperfect power factorization do not entail
integer factorization, but the radicands can become large. Moreover,
although input \#11 in Table \ref{Flo:AlternateAbSurds} is among
the most concise for this particular absurd number, this display form
would not be liked by many users on some other examples, and these
internal forms are not helpful as a point of departure for most of
the display forms that users might want. 
\end{itemize}
Thus, the advantages of proper primal internal form subsume those
of coprime square-free distinct exponent internal forms, and in comparison
the disadvantages of canonical forms based solely on perfect power
factorization do not make up for its better worst-case computing time.
These considerations suggest the following hybrid internal representation
and algorithmic ideas:
\begin{enumerate}
\item Use radicand factorization and proper-exponent primal form up through
some prime $\hat{p}$. Composite integer factors exceeding $\hat{p}^{2}$
are merely perfect-power factored, and the exponents are made proper.
Let us call any consequent integer radicand exceeding $\hat{p}^{2}$
a \textsl{quasi-prime}. Tight balanced exponents could be used instead
of proper exponents.
\item Treat the quasi-prime factors the same as prime factors, except compute
gcds between each new quasi prime and any other quasi primes and a
rational denominator and/or numerator whose magnitude exceeds $\hat{p}$.
Any resulting non-trivial gcd splits the radicand or radicands and
might enable extracting more rational coefficient. This process is
illustrated in computations (\ref{eq:GcdsCanHelp}) through (\ref{eq:UseGCDSTep4}).
\item The resulting form is not necessarily canonical if it contains a quasi-prime
radicand, which will be rather infrequent if $\hat{p}$ is set rather
large. However:

\begin{enumerate}
\item The radicands are always coprime to each other and the numerator and
denominator of any rational factor.
\item Addition, multiplication and rational powers of absurd numbers always
yield a single absurd number if the result can be so represented,
and therefore 0 is always recognized in such results.
\end{enumerate}
\end{enumerate}

\section{Conclusions\label{sec:Conclusion}}

Some major computer algebra systems currently produce erroneous results
that could be prevented by transforming absurd numbers to a canonical
internal form. There are reasonably efficient canonical forms that
avoid the potential cost of factoring large integers, and they are
not difficult to implement. However the radicands can become large,
the arithmetic is slower when there are no large prime factors, and
these forms are not good points of departure for popular display forms.

Thus for an \textsl{internal} representation we recommend using the
proper or tight balanced exponent primal form up to some particular
prime base $\hat{p}$, beyond which only perfect power factorization
is used, together with gcds to assure 0-recognition and that all factors
are coprime.

We think the default display form should be concise. No one form will
be the most concise for all examples, but the most concise will often
be in the set:
\begin{enumerate}
\item a rational number times one minimal integer raised to a proper exponent,
as exemplified by input \#15;
\item coprime square-free bases raised to distinct exponents \textbf{--}
perhaps times a rational number \textbf{--} as exemplified by inputs
\#5 through \#10.
\end{enumerate}
A system could compute all these forms and perhaps others, and then
display the most concise one. However, when displaying an expression
containing multiple absurd numbers, consistency in the form used for
those numbers is also important. Also the prime factorization provided
by the pure, loose balanced-exponent, and proper primal forms is particularly
informative. Therefore systems should provide a convenient mechanism
for users to set the default form used to display absurd numbers.
Ideally this default display form setting would be done using a transformation
dialog box such as that shown in Section 2.

\end{document}